\documentclass[a4paper,11pt]{article}
\usepackage{a4wide}
\usepackage{microtype}
\usepackage[hidelinks=true]{hyperref}
\usepackage{xcolor}
\usepackage{amsmath}
\usepackage{amssymb}
\usepackage{xspace}
\usepackage{complexity}
\usepackage[boxruled]{algorithm2e}
\usepackage{times}
\usepackage{amsthm}

\theoremstyle{plain}
\newtheorem{theorem}{Theorem}
\newtheorem{lemma}{Lemma}

\theoremstyle{remark}

\newtheorem{claim}{Claim}[lemma]

\theoremstyle{definition}
\newtheorem{definition}{Definition}

\newenvironment{subproof}{%
    \begin{proof}[Proof of Claim]%
    }{%
    \end{proof}%
}

\usepackage{enumitem}
\setlist{nolistsep, noitemsep, topsep=0pt}

\usepackage{tikz}
\usetikzlibrary{calc}
\usetikzlibrary{decorations.markings}
\usetikzlibrary{shapes.geometric,positioning}
\usetikzlibrary{arrows}
\usetikzlibrary{patterns}
\tikzset{node/.style={minimum size=1.8mm,circle,fill=black,draw,inner sep=0pt},
    decoration={markings,mark=at position .5 with {\arrow[black,thick]{stealth}}}}
\tikzset{req/.style={minimum size=1.8mm,circle,fill=white,draw,inner sep=0pt},
    decoration={markings,mark=at position .5 with {\arrow[black,thick]{stealth}}}}

\usepackage[noblocks]{authblk}
\title{A QPTAS for Gapless MEC\footnote{Research partially funded by Deutsche Forschungsgemeinschaft grant MO2889/1-1.}}
\author{Shilpa Garg}
\affil{Max Planck Institute for Informatics, Saarland Informatics Campus, Germany. \texttt{sgarg@mpi-inf.mpg.de}}
\author{Tobias M\"omke}
\affil{Saarland University, Saarland Informatics Campus, Germany and University of Bremen, Germany. \texttt{moemke@cs.uni-saarland.de}}

\newcommand{\ie}{i.e.}

\newcommand{\WLOG}{w.l.o.g.\xspace}
\newcommand{\cost}{\ensuremath{\mathrm{cost}}\xspace}
\newcommand{\MEC}{\textsc{MEC}\xspace}
\newcommand{\GMEC}{\textsc{Gapless-MEC}\xspace}
\newcommand{\BMEC}{\textsc{Binary-MEC}\xspace}

\newcommand{\dist}{\ensuremath{\mathrm{dist}}\xspace}

\newcommand{\euler}{\textrm{e}}
\newcommand{\SWC}{\textsc{SWC}_{\varepsilon^3}}

\newclass{\opt}{opt}
\newclass{\reference}{ref}
\newclass{\UG}{UG}
\newcommand{\st}{\textrm{start}}
\newcommand{\en}{\textrm{end}}

\usepackage{xcolor}

\usepackage{xspace}
\usepackage{complexity}
\usepackage[boxruled]{algorithm2e}
\usepackage[utf8]{inputenc}

\usepackage{tikz}
\usetikzlibrary{calc}
\usetikzlibrary{decorations.markings}
\usetikzlibrary{shapes.geometric,positioning}
\usetikzlibrary{arrows}
\usetikzlibrary{patterns}
\tikzset{node/.style={minimum size=1.8mm,circle,fill=black,draw,inner sep=0pt},
    decoration={markings,mark=at position .5 with {\arrow[black,thick]{stealth}}}}
\tikzset{req/.style={minimum size=1.8mm,circle,fill=white,draw,inner sep=0pt},
    decoration={markings,mark=at position .5 with {\arrow[black,thick]{stealth}}}}

\usepackage{cleveref}
\begin{document}
\maketitle

\begin{abstract}
    We consider the problem Minimum Error Correction (\MEC).
    A \MEC instance is an $n \times m$ matrix $M$ with entries from $\{0,1,-\}$. 
    Feasible solutions are composed of two binary $m$-bit strings, together with an assignment of each row of $M$ to one of the two strings.
    The objective is to minimize the number of mismatches (errors) where the row has a value that differs from the assigned solution string.
    The symbol ``$-$'' is a wildcard that matches both $0$ and $1$.
    A \MEC instance is gapless, if in each row of $M$ all binary entries are consecutive.

    \GMEC is a relevant problem in computational biology, and it is closely related to segmentation problems that were introduced by {[}Kleinberg--Papadimitriou--Raghavan STOC'98{]} in the context of data mining.

    Without restrictions, it is known to be $\UG$-hard to compute an $O(1)$-approximate solution to \MEC. For both \MEC and \GMEC, the best polynomial time approximation algorithm has a logarithmic performance guarantee.
    We partially settle the approximation status of \GMEC by providing a quasi-polynomial time approximation scheme (QPTAS).
    Additionally, for the relevant case where the binary part of a row is not contained in the binary part of another row, we provide a polynomial time approximation scheme (PTAS).
\end{abstract}

\section{Introduction.}
The minimum error correction problem (\MEC) is a segmentation problem where we have to partition a set of length $m$ strings into two classes.
A \MEC instance is given by a set of $n$ strings over $\{0,1,-\}$ of length $m$, where the symbol ``$-$'' is a wildcard symbol.
The strings are represented by an $n \times m$ matrix $M$, where the $i$th string determines the $i$th row $M_{i,*}$ of $M$.
The distance $\dist$ of two symbols $a,a'$ from $\{0,1,-\}$ is

\[
    \dist(a,a') := \left\{
        \begin{array}{ll}
            1\colon & a=0,a'=1 \mbox{ or } a=1, a'=0\\ 
            0\colon & \mbox{otherwise.}
        \end{array}
    \right.
\]

For two strings $s,s'$ from $\{0,1,-\}^m$ where $s_j, s'_j$ denotes the $j$-th symbol of the respective string, $\dist(s,s') := \sum_{j=1}^m \dist(s_j,s'_j)$.
A feasible solution to \MEC is a pair of two strings $\sigma,\sigma'$ from $\{0,1\}^m$.
The optimization goal is to find a feasible solution $(\sigma,\sigma')$ that minimizes

\[
    \cost_M(\sigma,\sigma') := \sum_{i=1}^{n} \min\{\dist(M_{i,*},\sigma), \dist(M_{i,*},\sigma')\}\,.
\] 

If $M$ is clear from the context, we sometimes skip the index.

A \MEC instance is called \emph{gapless} if in each of the $n$ rows of $M$, all entries from $\{0,1\}$ are consecutive. 
(As regular expression, a valid row is a word of length $m$ from the language $-^*\{0,1\}^*-^*$).
The \MEC problem restricted to gapless instances is \GMEC.

Our motivation to study \GMEC stems from its applications in computational biology.
Humans are diploid, and hence there exist two versions of each chromosome.
Determining the DNA sequences of these two chromosomal copies -- called haplotypes -- is important for many applications ranging from population history to clinical questions~\cite{snyder2015haplotype, tewhey2011importance}. 
Many important biological phenomena such as compound heterozygosity, allele-specific events like DNA methylation or gene expression can only be studied when haplotype-resolved genomes are available~\cite{leung2015integrative}. 

Existing sequencing technologies cannot read a chromosome from start to end, but instead deliver small pieces of the sequences (called reads). 
Like in a jigsaw puzzle, the underlying genome sequences are reconstructed from the reads by finding the overlaps between them.

The upcoming next-generation sequencing technologies (e.g., Pacific Biosciences) have made the production of relatively long contiguous sequences with sequencing errors feasible, where the sequences come from both copies of chromosome.
These sequences are aligned to a reference genome or to a structure called contig. 
We can formulate the result of this process as a \GMEC instance: the sequences are the contiguous strings and the contig determines the columns of the strings.

\GMEC is a generalization of a problem called \BMEC, the version of \MEC with only instances $M$ where all entries of $M$ are in $\{0,1\}$.
Finding an optimal solution to \BMEC is equivalent to solving the hypercube 2-segmentation problem (H2S) which was introduced by Kleinberg, Papadimitriou, and Raghavan~\cite{KPR98_segmentation,KPR04_segmentation} and which is known to be $\NP$-hard \cite{Fei14_np,KPR04_segmentation}.
The optimization version of \BMEC  differs from H2S in that we minimize the number of mismatches instead of maximizing the number of matches.
\BMEC allows for good approximations.
Ostravsky and Rabiny~\cite{OR02_polynomial} obtained a PTAS for \BMEC based on random embeddings.
Building on the work of Li et al.~\cite{LMW02_finding}, Jiao et al.~\cite{JXL04_k} presented a deterministic PTAS for \BMEC.

\GMEC was shown to be $\NP$-hard by Cilibrasi et al.~\cite{CIKT07_complexity}.\footnote{Their result predates the hardness result of Feige~\cite{Fei14_np} for H2S. The proof of the claimed $\NP$-hardness of H2S by Kleinberg, Papadimitriou, and Raghavan~\cite{KPR98_segmentation} was never published.}
Additionally, they showed that allowing a single gap in each string renders the problem $\APX$-hard.
More recently, Bonizzoni et al.~\cite{BDK+16_minimum} showed that it is unique games hard to approximate \MEC with constant performance guarantee, whereas it is approximable within a logarithmic factor in the size of the input. 
To our knowledge, previous to our result their logarithmic factor approximation was also the best known approximation algorithm for \GMEC.

\subsection{Our results.}
Our main result is the following theorem.
\begin{theorem}\label{thm:qptas}
    There is a quasi-polynomial time approximation scheme (QPTAS) for \GMEC.
\end{theorem}
Thus we partially settle the approximability for this problem: \GMEC is not $\APX$-hard unless $\NP \subseteq \QP$ (cf.~\cite{RS09_approximation}).
Thus our result reveals a separation of the hardness of the gapless case and the case where we allow a single gap.
Furthermore, already \BMEC is strongly $\NP$-hard since the input does not contain numerical values. 
Therefore we can exclude the existence of an FPTAS for both $\BMEC$ and $\GMEC$ unless $\P = \NP$.

Additionally, we address the class of \emph{subinterval-free} \GMEC instances where no string is contained in another string.
More precisely, for each pair of rows from $M$ we exclude that the set of columns with binary entries from one row is a strict subset of
the set of columns with binary entries from the other row.

\begin{theorem}\label{thm:ptas}
    There is a polynomial time approximation scheme (PTAS) for $\GMEC$ restricted to instances such that no string is the substring of another string.
\end{theorem}

\subsection{Overview of our approach.}

Our algorithm is a dynamic program (DP) that is composed of several levels.
Given a general \GMEC instance, we decompose the rows of the instance into length classes according to the length of the contiguous binary parts of the rows.
For each length class we consider a well-selected set of columns such that each row crosses at least one column and at most two.
(Row $i$ crosses a column $j$, if $M_{i,j} \in \{0,1\}$, i.e., the binary part of the row contains that column.)

We further decompose each length class into two sub-classes, one that crosses exactly one column and one that crosses exactly two columns.
For the second class, it is sufficient to consider every other column, which leaves us with many \emph{rooted} instances.
Thus for each sub-instance there is a single column (the root) which is crossed by all rows of the instance.

We further decompose rooted sub-instances into the left hand side and the right hand side of the root.
Since the two sides are symmetric, we can arrange the rows and columns of these sub-instances in such a way that all rows cross the first column.
We call this type of sub-instance \emph{SWC-instance} (for ``simple wildcards'').
We order the rows from top to bottom by increasing length in order to be able to further decompose the instance.

The first level of our DP solves these highly structured SWC-instances.
The basic idea that we would like to apply is that we select a constant number of rows from the instance that represents the solution.
Without further precautions, however, this strategy fails because of differing densities within the instance: 
the selected rows have to represent both the entries of columns crossed by many short rows and entries of arbitrarily small numbers of rows crossing many columns.
To resolve this issue, we observe that computing the solution strings $\sigma$ and $\sigma'$ is equivalent to finding a partition of $M$ into two row sets, one assigned to $\sigma$ and the other assigned to $\sigma'$.
If we assume to have the guarantee that for both solution strings $\sigma$ and $\sigma'$ an $\varepsilon$ fraction of rows of the matrix $M$ forms a \BMEC sub-instance, we show that the basic idea works.

This insight motivates to separate SWC-instances from left to right into sub-instances with the required property and to assemble them from left to right using a DP.
There are, however, several complications.
In order to choose the right sub-instances, we have to take into account that the choice depends on which rows are assigned to $\sigma$ and which are assigned to $\sigma'$.
Therefore the DP has to take special care when identifying the sub-instances.

Furthermore, in order to stitch sub-instances together to form a common solution, the solution computed in the left sub-instance has to compute a set of candidate solutions oblivious of the choices of the right sub-instance.
This means that we have to compute a solution to the left sub-instance without looking at a fraction of rows.
We present an algorithm for these sub-instances in Section~\ref{sec:swc}.

In order to combine the sub-instances, we face further technical complications due to having distinct sub-instances for those rows assigned to $\sigma$ and those rows assigned to $\sigma'$.
In Section~\ref{sec:second}, we introduce a DP whose DP cells are pairs of simpler DP cells, one for $\sigma$ and one for $\sigma'$.

Before we consider general instances, we first develop our techniques by considering subinterval-free instances which are easier to handle (Section~\ref{sec:subinterval-free}).
Observe that the instances considered until now are special rooted sub-interval-free instances.
We show how to solve arbitrary rooted sub-interval-free instances by combining the DP with additional information about the sub-problems that contain the root.
We then introduce the notion of domination in order to combine rooted sub-interval-free instances with a DP proceeding from left to right.
The main idea is that a dominant sub-problem dictates the solution.
At the interface of two sub-instances, there can be a (contiguous) region where none of the two sub-problems is dominant.
We show that these regions can be solved directly by considering a constant number of rows (using the results from Section~\ref{sec:swc}).

Until this point, all parts of our algorithm run in polynomial time.
We lose this property when considering length classes, in Section~\ref{sec:length-classes}.
The length classes allow us to separate an instance into rooted sub-instances.
The difficulty is that the left hand side of a separating column may have a completely different structure than the right hand side of that column.
We do not know how to combining the two sides by considering only a polynomial number of possibilities.
If we allow, however, quasipolynomial running time, we can solve the problem. 
We use that each of the two sub-instances (the one on the left and the one on the right) is composed of at most logarithmically many parts.
Considering all parts simultaneously allows us to take care of dependencies between the left hand side and the right hand side and still solve them as if they were separate instances.

Combining such rooted instances from left to right then can be done in the same spirit as combining rooted sub-interval-free instances.
To solve the entire length-class, we combine both solutions by running a new DP that considers quadruples of DP cells.

Finally, in Section~\ref{sec:generalQPTAS}, we are able to handle all length classes simultaneously. 
We solve general instances in the same spirit as the combined sub-instances of a single length class.
Instead of considering quadruples of cells, however, we form collections of quadruples that are -- figuratively speaking -- stacked on top of each other.
The key insight is that there are only $O(\log(n))$ different length classes and each collection has at most one quadruple of each length class.
Considering all possible collections adds another power of $\log(n)$ to the running time, which is still quasi-polynomial.

\subsection{Further related work.}
Binary-MEC is a variant of the Hamming $k$-Median Clustering Problem when $k = 2$ and there are PTASs known~\cite{JXL04_k, OR02_polynomial}. 
Li, Ma, and Wang~\cite{LMW02_finding} provided a PTAS for the general consensus pattern problem which is closely related to \MEC.
Additionally, they provided a PTAS for a restricted version of the star alignment problem aligning with at most a constant number of gaps in each sequence.

Alon and Sudakov~\cite{AS99_two} provided a PTAS for H2S, the maximization version of \BMEC and Wulff, Urner and Ben-David~\cite{WUB13_monochromatic} showed that there is also a PTAS for the maximization version of \MEC.
For \GMEC, He et al.~\cite{HCP+10_optimal} studied the fixed-parameter tractability in the parameter of fragment length with some restrictions.
These restrictions allow their dynamic programming algorithm to focus on the reconstruction of a single haplotype and, hence, to limit the possible combinations for each column.
There is an FPT algorithm parameterized by the coverage~\cite{PMP+15_whatshap,garg2016read} (and some additional parameters for pedigrees). 
Bonizzoni et al.~\cite{BDK+16_minimum} provided FPT algorithms parameterized by the fragment length and the total number of corrections for Gapless-MEC.
There are some tools which can be used in practice to solve Gapless-MEC instances~\cite{PZD+15_hapcol, PMP+15_whatshap}.

Most research in haplotype phasing deals with exact and heuristic approaches to solve \BMEC.
Exact approaches, which solve the problem optimally, include integer linear programming~\cite{FM12_solving} and fixed-parameter tractable algorithms~\cite{HCP+10_optimal, PZD+15_hapcol}.
There is a greedy heuristic approach proposed to solve Binary-MEC~\cite{BB08_hapcut}. 

Lancia et al.~\cite{LBI+01_snps} obtained a network-flow based polynomial time algorithm for Minimum Fragment Removal (MFR) for gapless fragments.
Additionally, they found the relation of Minimum SNPs Removal (MSR) to finding the largest independent set in a weakly triangulated graph.

\subsection{Preliminaries and notation.}\label{sec:prelim}

We consider a \GMEC instance, which is a matrix $M \in \{0,1, -\}^{n \times m}$.
The $i$th row of $M$ is the vector $M_{i,*} \in \{0,1, -\}^{1 \times m}$ and the $j$th column is the vector $M_{*,j} \in \{0,1, -\}^{n \times 1}$.
The length of the binary part in $M_{i,*}$ is $|M_{i,*}|$. 
We say that the $i$th row of $M$ \emph{crosses} the $j$th column if $M_{i,j} \in \{0,1\}$.

For each feasible solution $(\sigma,\sigma')$ for $M$, we specify an assignment of rows $M_{i,*}$ to solution strings.
The default assignment is specified as follows.
For a row $M_{i,*}$, we assign $M_{i,*}$ to $\sigma$ if $\dist(\sigma,M_{i,*}) \le \dist(\sigma',M_{i,*})$.
Otherwise we assign $M_{i,*}$ to $\sigma'$.
For the rows of $M$ assigned to $\sigma$ we write $\sigma(M)$ and for the rows assigned to $\sigma'$ we write $\sigma'(M)$.
For a given instance, $\Opt = (\tau, \tau')$ denotes an optimal solution.
Observe that knowing $\Opt$ allows us to obtain an optimal assignments $\tau(M)$ and $\tau'(M)$ by assigning each row to the solution string with fewest errors and knowing $\tau(M)$ and $\tau'(M)$ allows us to obtain an optimal solution by selecting the column-wise majority values.

\section{Simple instances with wildcards.}\label{sec:swc}
In this section, we consider instances of \GMEC where all entries of column one in $M$ are zero or one, i.e., $M_{i,1} \in \{0,1\}$ for each index $i$.
Observe that the wildcards now have a simple structure which we refer to as SWC-structure.
An instance with SWC-structure is an SWC-instance.

\begin{definition}[Standard ordering of SWC-instances]
    We define the \emph{standard ordering} of rows in $M$ such that $|M_{i,*}| \le |M_{i+1,*}|$ for each $i$, \ie, we order them from top to bottom in increasing length of the binary part. 
    \label{def:order-SWC}
\end{definition}

\begin{definition}[Good SWC-instances]
    \label{def:good-SWC}
    We call an SWC-instance $M$ \emph{good}, if it is in standard ordering and there are at least $\varepsilon |\tau(M)|$ rows of $\tau(M)$ and at least $\varepsilon |\tau'(M)|$ rows of $\tau'(M)$ that have only entries from $\{0,1\}$.
\end{definition}

To solve good SWC-instances, we generalize the PTAS for \BMEC by Jiao et al.~\cite{JXL04_k}. 
Our algorithm requires partitions of the set of rows.
In the following two definitions, the required number of rows may be a fractional number. 
To solve the problem, we allow the assignment of fractional rows, i.e., for a row $i$, we can choose an $x \in [0,1]$ and assign an $x$ fraction of $i$ to one set and a $1-x$ fraction to the other set. 

The following two definitions allow us introduce a structured view on optimal solutions.
\begin{definition}[Trisection]
    An \emph{$\varepsilon$-trisection} of an instance $M$ for $\tau$ is a partition of the rows into three consecutive ranges that have the following properties.
    \begin{enumerate}
        \item The first range $U$ contains row $M_{1,*}$ and $(1-\varepsilon) |\tau(M)|$ rows of $\tau(M)$.
        \item The second range $L$ is consecutive to first row set containing $(\varepsilon - \varepsilon^2)|\tau(M)|$ rows of $\tau(M)$.
        \item The third range $X$ contains the remaining rows in $M$.
    \end{enumerate}
    To avoid ambiguity, we choose $L$ and $X$ such that the first row is in $\tau(M)$.

    We define an $\varepsilon$-trisection $U'$, $L'$, and $X'$ for $\tau'$ analogously, replacing $\tau(M)$ by $\tau'(M)$.
    \label{def:trisection}
\end{definition}

\begin{definition}[Subdivision of trisections]
    We consider the rows sets $U, L, U', L'$ from Definition~\ref{def:trisection} and additionally, we divide each of these sets into $1/\varepsilon^2$ disjoint subsets denoted as $U_i, L_i, U'_i, L'_i$.
    For each $i$, $U_i$ contains $\varepsilon^2 \cdot |U|$ rows from $\tau(M)$ and $L_i$ contains $\varepsilon^2 \cdot |L|$ rows from $\tau(M)$.
    Analogously, each $U'_i$ contains $\varepsilon^2 \cdot |U'|$ rows from $\tau'(M)$ and $L'_i$ contains $\varepsilon^2 \cdot |L'|$ rows from $\tau'(M)$.
    To avoid ambiguity, each set $U_i$ and $L_i$ starts with a (fractional) row of $\tau(M)$ and each set $U'_i$ and $L'_i$ starts with a (fractional) row of $\tau'(M)$. 
    \label{def:subdivision}
\end{definition}

We introduce a new algorithm $\textsc{SWC}_\delta$ for our setting.
For an instance $M$, we consider the rows sets $U, L, U', L'$ from the $\varepsilon$-trisections of $M$ and their subsets according to Definition~\ref{def:subdivision}.
Additionally, we select a multi-set of rows from $U'_i \cap \tau'(M)$ and $L'_i \cap \tau'(M)$.
We then compute the majority weighting according to Definition~\ref{def:weighted-majority} for each column $j$ using 
multisets based on the minimum number of errors.
The main idea to find two small row sets that represent the whole instance $M$.
The intuitive meaning is that we select rows from the upper part with a much lower density then the rows of the lower part.
We therefore introduce a bias such that all rows are equally important.

\begin{definition}[Weighted majority]
    Let $j$ be an integer and let $\tilde{U}$ and $\tilde{L}$ be two matrices with at least $j$ columns.
    In $\tilde{U}_{*,j}$ and $\tilde{L}_{*,j}$, we replace all zeros by $-1$ and then all wildcard symbols by zero.
    We then compute the number
    $\nu := \sum_{i \in \tilde{U}_{i,j}} (1-\varepsilon)i/(\varepsilon - \varepsilon^2) + \sum_{i \in \tilde{L}_{i,j}} i$. 
    Then $\textsc{Majority}_j(\tilde{U}, \tilde{L}) = 0$ if $\nu <0$ and $\textsc{Majority}_j(\tilde{U}, \tilde{L}) = 1$ if $\nu \ge 0$.
    \label{def:weighted-majority}
\end{definition}
With this preparation, we are now ready to present the algorithm.
The input has a long list of parameters that will allow our dynamic programs later on to control the execution.
The reason is that we do not know $\tau$ and $\tau'$. 
Therefore the algorithm takes \emph{guesses} of row sets as input.
The values $r$ and $r'$ are guesses of $|\tau(M)|$ and $|\tau'(M)|$.
\smallskip

\begin{algorithm}[H] 
    \caption{
        \label{alg:SWC}
        $\textsc{SWC}_\delta$}
    \SetAlgoNoLine
    \SetNlSkip{1em}
    \SetKwInOut{Input}{Input}
    \SetKwInOut{Output}{Output}
    \Input{Row sets $U_i$, $L_i$, $U'_i$ and $L'_i$ of a good SWC-instance $M$,  numbers $r, r'$.\\ Optional: selection of rows $\tilde{U}_i,\tilde{L}_i,\tilde{U}'_i,\tilde{L}'_i$, see below.}
    \Output{A pair of solution strings $(\sigma,\sigma')$.}
    Run the algorithm for each possible selection of the following type and keep the best outcome (minimum number of errors)\tcp*{If provided as input, skip selection.}
    For each $i$, select (with repetition) a multi-set $\tilde{U}_i$ of $1/\delta$ rows from $U_i$ and $\tilde{L}_i$ from $L_i$\;
    For each $i$, select (with repetition) a multi-set $\tilde{U}'_i$ of $1/\delta$ rows from $U'_i$ and $\tilde{L}'_i$ from $L'_i$ such that $\tilde{U}' \cap \tilde{U} = \tilde{L}' \cap \tilde{L} = \emptyset$\;
    \tcp{$\tilde{U} := \bigcup_i \tilde{U_i}$. The values $\tilde{U}'$, $\tilde{L}$, and $\tilde{L}'$ are defined analogously.}
    For each column $j$, set $\sigma_j := \textsc{Majority}_j(\tilde{U}, \tilde{L})$ and $\sigma'_j := \textsc{Majority}_j(\tilde{U}', \tilde{L}')$\;
    For each row $i$ of $M$, determine the value $d_i := \dist(\sigma,M_{i,*}) - \dist(\sigma',M_{i,*})$\;
    Assign the $r$ rows with minimal values $d_i$ to $\sigma$ and the remaining $r'$ rows to $\sigma'$.
\end{algorithm}
\smallskip

Observe that for small (\ie, constant) values of $r$ or $r'$, the algorithm $\textsc{SWC}_\delta$ can be replaced by an exact algorithm since we know $\tau(M)$ if and only if we know $\tau'(M)$, and we are able to guess constantly many rows.

\begin{lemma}\label{lem:SWC}
    Let $M$ be a good SWC-instance.
    For sufficiently large $r = |\tau(M)|$ and $r' = |\tau'(M)|$, let $U_i, L_i, U'_i, L'_i$ be a subdivision (Definition~\ref{def:subdivision}) of an $\varepsilon$-trisection $U,L,X,U',L',X'$ of $M$.
    Then $\textsc{SWC}_{\varepsilon^3}$ is a $(1 + O(\varepsilon))$-approximation algorithm for $M$. 
\end{lemma}
The proof is based on a randomized argument using Chernoff bounds. (See Appendix~\ref{app:SWC}).

Lemma~\ref{lem:SWC} shows that the set of solutions considered by $\SWC$ contains at least one solution that is good enough even though we do not look at $X$. 
It does not say that we finally compute that solution, since other solutions may have fewer errors in $U\cup L$ or $U' \cup L'$.
For our dynamic programs, we need a stronger statement.
We would like to be able to compute a solution for an instance and \emph{afterwards} change a fraction of assignments without losing the approximation guarantee.
The next lemma is a key ingredient of our result.

\begin{lemma}
    \label{lem:swc-gap}
    Let $M$ be a good SWC-instance and $\varepsilon > 0$ sufficiently small.
    Let $U,L,X$ be an $\varepsilon$-trisection for $\tau(M)$ and $U',L',X'$ an $\varepsilon$-trisection for $\tau'$, with subdivisions $U_i,L_i,U'_i,L'_i$ according to Definition~\ref{def:subdivision}.
    Let $(\sigma, \sigma')$ be the solution computed by $\textsc{SWC}_{\varepsilon^3}$ with $r = |\tau(M)|$, $r' = |\tau'(M)|$.
    Then re-assigning the rows $\sigma(X)$ to $\tau(X)$ and $\sigma'(X')$ to $\tau'(X')$ gives a $(1 + O(\varepsilon))$-approximation for the instance $M$.
\end{lemma}
\begin{proof}
    For ease of presentation, we assume that all appearing numbers are integers.
    It is easy to adapt the proof by rounding fractional numbers appropriately.

    We first analyze the computed solution string $\sigma$. 
    Let $\eta$ be the total number of errors of $(\tau,\tau')$ within $M$ and
    let $\eta_P$ be the total number of errors of $(\sigma,\sigma')$ within $P := U \cup L$.
    Due to Lemma~\ref{lem:SWC}, we have $\eta_P \le (1+O(\varepsilon)) \eta$.

    We may assume $r \ge r'$ since otherwise we can simply rename the two strings $\tau$, $\tau'$.
    Additionally, by renaming of $\sigma$ and $\sigma'$, we may assume that $|\sigma(P) \cap \tau(P)| \ge |\sigma'(P) \cap \tau(P)|$.
    Therefore $|\tau(P)| \ge n/3$ and $|\sigma(P) \cap \tau(P)| \ge n/6$.
    (Recall that the matrix $M$ has $n$ rows and $m$ columns. The value $n/3$ is a save bound on $n/2 - \varepsilon^2 n$.)

    \begin{claim}\label{claim:agree}
        There is a set $I$ of $m - 25\eta/n$ indices $j$ such that
        $\sigma_j = \tau_j$ for all $j \in I$.
    \end{claim}
    \begin{subproof}
        We concentrate on the columns of $M$ where both strings $\tau$ and $\sigma$ have at most $n/12$ errors within $P$.
        By counting the errors, there are at most $12 \eta/n$ columns where $\tau$ has at least $n/12$ errors.
        Similarly, there are at most $12 (1+O(\varepsilon)) \eta_P/ n < 13 \eta/ n$ many columns where $\sigma$ has at least $n/12$ errors.
        Therefore there is a set $I$ of at least $m - 25\eta/n$ columns where simultaneously both $\tau$ and $\sigma$ have less than $n/12$ errors each.

        Now suppose that the claim was not true and there was an index $j \in I$ with $\tau_j \neq \sigma_j$.
        Then, since $|\tau(P) \cap \sigma(P)| \ge n/6$, either $\sigma_j$ or $\tau_j$ is erroneous in at least $n/12$ rows of $\tau(P) \cap \sigma(P)$, a contradiction.
    \end{subproof}

    Next we analyze $\sigma'$ for the columns $I$.
    Let $j$ be a column (\ie, an index) from $I$.
    By symmetry, we may assume $\sigma_j = \tau_j = 0$. 
    We aim to show that an optimal solution has always sufficiently many errors to pay for wrong entries of $\sigma'$.

    Let $\eta_j$ be the number of errors of $(\tau,\tau')$ in column $j$ of $M$ and
    let $\eta_{P,j}$ be the number of errors of $(\sigma,\sigma')$ in column $j$ of $P$.
    Let $\eta''_j = \eta_j + \eta_{P,j}$.
    \begin{claim}\label{claim:errors}
        For each column $j$ of $I$, either $\sigma'_j = \tau'_j$ or $\eta''_j \ge (\varepsilon-\varepsilon^2)|\tau'(M)|/2$.
    \end{claim}
    \begin{subproof}
        We distinguish two cases.
        We first assume $\tau'_j = 0$.
        If also $\sigma'_j = 0$, we are done. 
        We therefore assume $\sigma'_j = 1$.
        If there are more than $|\tau'(L')|/2$ ones in column $j$ of $L'$, $(\tau,\tau')$ has more than $|\tau'(L')|/2$ errors in column $j$ and thus 
        $\eta_j \ge |\tau'(L')|/2$.
        Otherwise $\sigma'(L')$ has at least $|\tau'(L')|/2$ zeros in column $j$ and therefore $\eta_{P,j} \ge |\tau'(L')|/2$.
        We obtain $\eta''_j \ge |\tau'(L')|/2 \ge (\varepsilon-\varepsilon^2) |\tau'(M)|/2$ as claimed.

        In the second case, $\tau'_j = 1$ and we assume that $\sigma'_j = 0$.
        If there are more than $r'/2$ ones in column $j$ of $U'$, $(\sigma,\sigma')$ has more than $r'/2$ errors in column $j$ and thus $\eta_{P,j} \ge |\tau'(U')|/2$.
        Otherwise $\tau'(U')$ has at least $r'/2$ zeros in column $j$ and therefore $\eta_j \ge |\tau'(U')|/2$.
        Again, we obtain $\eta''_j \ge |\tau'(U')|/2 \ge (1-\varepsilon) |\tau'(M)|/2$ as claimed.
    \end{subproof}

    Since by our assumption $|\tau'(X')| < \varepsilon^2 |\tau'(M)|$,
    Claim~\ref{claim:errors} implies that within $I$, after reassigning the rows we still have a $(1+ O(\varepsilon))$-approximation.

    To finish the proof, we argue that $\eta$ is large enough to pay for all errors in $X$ and $X'$ outside of $I$.
    Let $\eta_I$ be the number of errors due to assigning $\sigma$ to $\tau(X)$ and $\sigma'$ to $\tau'(X')$ within the interval $I$.

    Then, using the size of $I$ stated in Claim~\ref{claim:agree}, the total number of errors of $(\sigma,\sigma')$ in $M$ is at most
    $(1+O(\varepsilon)) \eta + \eta_I + \varepsilon^2 n \cdot 25\eta/n$, \ie, the errors of $\textsc{SWC}_{\varepsilon^3}$ within $P$, the errors within $X$ and $X'$ in the columns of $I$, and all other entries of $X \cup X'$.
    The obtained approximation ratio is

    $((1+O(\varepsilon)) \eta + \eta_I + \varepsilon^2 n \cdot 25\eta/n/\eta 
    \le (\eta + O(\varepsilon) \eta + 25 \varepsilon^2 \eta)/\eta
    = 1 +  O(\varepsilon)$.

    The first inequality uses that for some constant $k$, $(1 + k \varepsilon) \eta \ge \eta + \eta_I$.
\end{proof}

\subsection{A DP for SWC-instances.}
\label{sec:second}

Let $M$ be an SWC-instance with rows  $\{1, 2, \ldots n\}$. 
We define $\st_i$ to be the start and $\en_i$ the end of string number $i$ of $M$, i.e., the column number of the matrix where the binary part starts and ends.
For a sub-matrix $M'$ of $M$, $\st_{M'}$ determines the index of the first column of $M'$ and $\en_{M'}$ the index of the last column of $M'$.

We next specify the parts of which the DP cells are composed.
We divide the input instance into blocks defined as follows.
\begin{definition}[Block]
    Given a good SWC-instance $M$, a block $B$ is a sub-instance determined by three numbers $1 \le a < b < c \le n$ as follows.
    The first column of $B$ is column $1$ of $M$. 
    The last column of $B$ is $\en_{b}$.
    The first row of $B$ is $a$ and the last row is $n$.
    We write $U_B$ for the rows from $a$ to $b - 1$, $L_B$ for the rows from $b$ to $c - 1$, and $X_B$ for the rows from $c$ to $n$. 
    \label{def:sets-one-solution-string}
\end{definition}
The idea is that a block determines a trisection. 
We subdivide each block into chunks and select rows from these chunks. 
Chunks are closely related to subdivisions of trisections, but we do not assume the knowledge of $(\tau,\tau')$.
\begin{definition}[Chunk] 
    Let $B$ be a block determined by the numbers $a,b,c$.
    We partition $B$ into $2/\varepsilon^2$ many \emph{chunks} (ranges or rows).
    These chunks are determined by numbers 
    \[
    a = a_1 < a_{2} < \dotsm < a_{1/\varepsilon^2 + 1} = b = b_{1} < b_{2} < \dotsm < b_{1/\varepsilon^2 + 1} = c\,.
    \]
    The $\ell$th chunk of $U_B$ is the submatrix composed of the rows $a_{\ell}$ to $a_{\ell+1}-1$ and the $\ell$th chunk of $L_B$ is the submatrix composed of the rows $b_{\ell}$ to $b_{\ell+1}$.
    \label{def:subsets-one-solution-string}
\end{definition}

\begin{definition}[Selection]
    \label{def:selection}
    For each block $B$ with a set of chunks $C$, we consider multiset $T$ of rows of size $2/\varepsilon^5$.
    We require that $T$ contains $1/\varepsilon^3$ rows from each chunk in $C$.
\end{definition}
The selection $T$ will take the role of $\tilde{U}$ and $\tilde{L}$ in $\textsc{SWC}_\delta$.
With these preparations we can define a DP cell.
\begin{definition}[DP cell]
    For each block $B$, each set of chunks $C$ of $B$ and each selection $T$ of rows from $B$, there is a DP cell represented by $D(B,C,T)$. 
    A DP cell $D(B,C,T)$ is a \emph{predecessor} of $D(\hat B,\hat C,\hat T)$ if the following conditions hold.
    \begin{itemize}
        \item $\hat{a} = b$ and $\hat{b} = c$, where $b,c,\hat{a},\hat{b}$ are the numbers from Definition~\ref{def:sets-one-solution-string}.
        \item The chunks from $C$ between $b$ and $c$ are exactly the chunks from $\hat{C}$ between $\hat{a}$ to $\hat{b}$.
        \item For each pair of chunks from $T \times \hat{T}$ with the same range of rows, the selection $T$ matches the selection $\hat{T}$. 
            \label{def:dp-cell-one-solution-string}
    \end{itemize}
\end{definition} 

The value of $D(B,C,T)$ will be an approximation of the minimum number of errors
that we can have in $M$ until the last column of $B$.

We now describe the dynamic program for a pair of solution strings $(\sigma, \sigma')$ by using joint DP cells $(\zeta,\zeta')$ (see also Fig.~\ref{fig:DP-crux1}). 

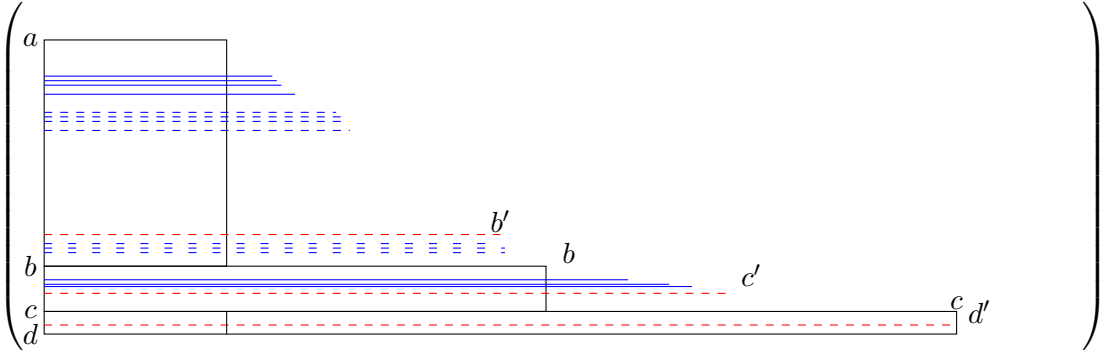
\begin{figure}[h]
    \begin{center}
        \begin{tikzpicture}[scale=0.6]
            \draw (-2,-2) rectangle (2,3);      
            \node at (-2.3,3) {$a$};
            \draw (-2,-3) rectangle (9,-2);     
            \node at (-2.3,-2) {$b$};
            \node at (-2.3,-3) {$c$};
            \node at (-2.3,-3.5) {$d$};
            \draw[red,dashed] (-2,-1.3) -- (8,-1.3);
            \node at (8, -1) {$b'$};
            \node at (9.5, -1.7) {$b$};
            \node at (13.5, -2.2) {$c'$};
            \draw (-2,-3.5) rectangle (18,-3);
            \draw[red,dashed] (-2,-2.6) -- (13,-2.6);
            \draw[red,dashed] (-2,-3.3) -- (18,-3.3);
            \node at (18.5, -3) {$d'$};
            \node at (18, -2.8) {$c$};
            \node at (-2.7,0){\(\left(\rule{0cm}{2.4cm}\right.\)};
            \node at (21,0){\(\left.\rule{0cm}{2.4cm}\right)\)};
            \draw (2,-3) to (2,-3.5);
            \draw[blue,dashed] (-2,1) -- (4.7,1);
            \draw[blue,dashed] (-2,1.2) -- (4.6,1.2);
            \draw[blue,dashed] (-2,1.3) -- (4.5,1.3);
            \draw[blue,dashed] (-2,1.4) -- (4.4,1.4);
            \draw[blue] (-2,1.8) -- (3.5,1.8);
            \draw[blue] (-2,2) -- (3.2,2);
            \draw[blue] (-2,2.1) -- (3.1,2.1);
            \draw[blue] (-2,2.2) -- (3,2.2);
            \draw[blue] (-2,-2.3) -- (10.8,-2.3);
            \draw[blue] (-2,-2.4) -- (11.7,-2.4);
            \draw[blue] (-2,-2.45) -- (12.2,-2.45);
            \draw[blue,loosely dashed] (-2,-1.5) -- (8,-1.5);
            \draw[blue,loosely dashed] (-2,-1.6) -- (8.1,-1.6);
            \draw[blue, loosely dashed] (-2,-1.7) -- (8.1,-1.7);
        \end{tikzpicture}
        \caption{\label{fig:DP-crux1} Example for a pair of strings with $b' > b$. The blue lines and dashed blue ones represent sets $T$ and $T'$, and $T \cap T' = \emptyset$.}
    \end{center}
\end{figure}

For $\sigma'$, we use the same notation as in Definitions~\ref{def:sets-one-solution-string}, \ref{def:subsets-one-solution-string} and \ref{def:selection}, but we use the symbol prime (\,$\cdot '$\,) for all occurring variables.

\begin{definition}[DP cell for a pair]
    We define joint DP cell $(\zeta,\zeta') = (D(B,C,T), D'(B',C',T'))$ with the two single cells defined as in Definition~\ref{def:dp-cell-one-solution-string}. 
    We require that
    \begin{itemize}
        \item the rows of $C$ and $C'$ where chunks start are pairwise distinct, and
        \item $T \cap T' = \emptyset$.
    \end{itemize}
    \label{def:dp-cell-two-solution-string}
\end{definition}
\smallskip

\begin{definition}[Predecessor of a joint DP cell]
    A DP cell $(\hat{\zeta}, \hat{\zeta}')$ is a \emph{predecessor} of $(\zeta, \zeta')$ if (i) $\hat{\zeta} = \zeta$ and $\hat{\zeta}'$ is a predecessor of $\zeta'$; or (ii) $\hat{\zeta}$ is a predecessor of $\zeta$ and $\hat{\zeta}' = \zeta'$.
    \label{def:predecessor-two}
\end{definition}

\subparagraph{Algorithm ($\textsc{SWC}^{\sigma, \sigma'}$).}
The general idea of the algorithm is to guess trisections.
Suppose we initially chose blocks $B,B'$ that are the left-most trisections for $\tau$ and $\tau'$.
Then we obtain an approximation of the prefix of $(\tau,\tau')$ restricted to $B,B'$ (whichever ends first) by sampling rows of $U_B$,$L_B$,$U_{B'}$, and $L_{B'}$.
The sampled rows for $L_B$ and $L_{B'}$ provide the interface to the next step.
Suppose $L_{B'}$ starts at an earlier row than $L_{B}$. 
Then we guess the trisection of $M$ for $\tau$ restricted to the rows of $L_{B'}$ and $X_{B'}$.
Let $B''$ be that block of our algorithm.
Then $U_{B''} = L_B$ and we sample rows of $L_{B''}$ in order to approximate a new infix of $\tau$.
For a simplified version of the DP without the complications due to having two solution strings, we refer to Appendix~\ref{sec:single}.

We globally guess two numbers $r$ and $r'$ that represent $|\tau(M)|$ and $|\tau'(M)|$.
We split the processing into an initialization phase and an update phase.
In the initialization phase, we assign values to each DP cell  $(\zeta,\zeta')$ based on $\SWC$ with the following parameters.
We obtain $U_i, L_i$ from the chunks $C$ and $U'_i,L'_i$ from the chunks $C'$.
In the execution of $\SWC$, we use the selections $T,T'$ instead of trying all possible selections, i.e., $T$ and $T'$ determine all $\tilde{U}_i$, $\tilde{L}_i$, $\tilde{U}'_i$, and $\tilde{L}'_i$ in the algorithm.
Let $\tilde{B}$ be the matrix with rows from $1$ to the $\min\{c - 1,c' - 1\}$ and columns one to $\min\{\en_B, \en_{B'}\}$.
The solution of the computation is a pair of strings $(\sigma_{\zeta,\zeta'},\sigma'_{\zeta,\zeta'})$, the prefixes of the two computed strings until $\en_{\tilde{B}}$.
The value of $(\zeta,\zeta')$ is $\cost_{\tilde{B}}(\sigma_{\zeta,\zeta'},\sigma'_{\zeta,\zeta'})$.

In the update phase, we compute the value and the pair of strings of the DP cell $(\zeta,\zeta')$ as follows. 
We inductively assume that all DP cells for predecessors of $(\zeta,\zeta')$ have been updated already.
We try all predecessor pairs of DP cells and keep the one that gives the best result (see also Appendix~\ref{sec:single}).
Let $(\overline{\zeta},\overline{\zeta}')$ be a predecessor of $(\zeta,\zeta')$.
By symmetry, we assume without loss of generality that $b' < b$. 
There are two cases how the two pairs interact.
The first case is $\zeta = \overline{\zeta}$.
We run $\SWC$ on the columns $\en_{\overline{B}'}+1$ to $\en_{B}$ with the parameters from $(\zeta,\zeta')$ (see initialization).
To obtain the full solution, we append the computed string for $B'$ to the string $\sigma'_{\zeta,\overline{\zeta}'}$ (which is one of the solution strings of the predecessor pair).
Let $\tilde{B}$ be the matrix with rows from $1$ to the $\min\{c - 1,c' - 1\}$ and columns one to $\en_{B'}$.
The solution of the computation is a pair of strings $(\sigma_{\zeta,\zeta'},\sigma'_{\zeta,\zeta'})$, the prefixes of the two computed strings from column one to $\en_{\tilde{B}}$.
The potential new value of $(\zeta,\zeta')$ is $\cost_{\tilde{B}}(\sigma_{\zeta,\zeta'},\sigma'_{\zeta,\zeta'})$.
We replace the stored solution with the potential new solution if the cost has decreased.

The second case is $\zeta' = \overline{\zeta}$.
This case is the crux of the joint DP, since we have a ``switch'' of the role of $\sigma$ and $\sigma'$.

We run $\SWC$ on the columns $\en_{\overline{B}}$ to $\en_{B'}$ with the parameters from $(\zeta,\zeta')$ (see initialization).
To obtain the full solution, we then append the computed string for $B$ to the string $\sigma_{\overline{\zeta},\zeta'}$ (which is one of the solution strings of the predecessor pair).
Let $\tilde{B}$ be the matrix with rows from $1$ to the $\min\{c - 1,c' - 1\}$ and columns one to $\en_{B'}$.
The solution of the computation is a pair of strings $(\sigma_{\zeta,\zeta'},\sigma'_{\zeta,\zeta'})$, the prefixes of the two computed strings until $\en_{\tilde{B}'}$.
The potential new value of $(\zeta,\zeta')$ is $\cost_{\tilde{B}}(\sigma_{\zeta,\zeta'},\sigma'_{\zeta,\zeta'})$.
We replace the stored solution with the potential new solution if the cost has decreased.

For the last strings,
we additionally consider special cells that are defined as before, but with $c=n$ or $c'=n$. 
Intuitively, we use these cells when only at most $1/\varepsilon^4$ rows of $\tau(M)$ or $\tau'(M)$ are left.
For pairs of cells containing such $\zeta$ or $\zeta'$, our computation considers the optimal solution within the computation instead of $\SWC$.

\begin{theorem}\label{thm:column1}
    The algorithm $\textsc{SWC}^{\sigma,\sigma'}$ is a PTAS for $SWC$-instances.\\
\end{theorem}
\begin{proof}
    To see that the DP works in polynomial time, we observe that instead of simple DP cells in Lemma~\ref{lem:simpleDP} here we consider pairs of DP cells.
    Therefore the number of cells is squared and thus stays polynomial.
    During the recursive construction of the solution, we compare each cell to be computed with one compatible cell at a time.
    Therefore the construction of the solution also takes only polynomial time.
    As in Lemma~\ref{lem:simpleDP}, the computed solution is vacuously feasible.

    We continue with analyzing the quality of the computed solution.
    Let $(\tau,\tau')$ be an optimal solution.
    We set $r = |\tau(M)|$ and $r' := |\tau'(M)|$.
    By renaming the two strings we may assume that the last row of the first $(1-\varepsilon^2) r$ rows of $\tau(M)$ is below the first row of the last $\varepsilon^2 r'$ rows of $\tau'(M)$.

    We consider DP cells similar to the proof of Lemma~\ref{lem:simpleDP}.
    Starting from the top-most row of $\tau(M)$, for each $i \ge 0$, the $i$th range $Y_i$ contains the next 
    $(\varepsilon^{2i} - \varepsilon^{2i+2})r$ rows of $\tau(M)$.
    We assign the rows not in $\tau(M)$ such that the first row of each $Y_i$ is contained in $\tau(M)$.
    Then we choose $Y_i$ such that all rows of $M$ until $Y_{i+1}$ are contained in $Y_i$.

    We consider the DP cells $\zeta_i$ for each $i$ with the parameters $B_i$, $C_i$, and $T_i$.
    The block $B_0$ contains the rows of $Y_0$ and $Y_1$, and the columns one to the end of the first row of $\tau(B_0)$.
    For each $i > 0$, block $B_i$ contains the rows of $Y_i$ and $Y_{i+1}$, and the columns after those of $B_{i-1}$ to the end of the first row of $B_i$.

    If only a constant number of rows of $\sigma(M)$ are left, we can compute the partial solutions optimally and
    there are DP cells for exactly this purpose:
    there is a DP cell $\zeta_i$ such that the last $2/\varepsilon^5$ rows of $\tau(M)$ are located between $a_i$ and $c_i$ and $Y_i$ contains exactly these rows.
    As before, to keep a clean notation, in the following we implicitly assume that cells with constantly many rows of $\sigma(M)$ are handled separately.

    The chunks of $C_i$ are the ranges that equally distribute $\tau(B)$. The selection $T_i$ is the best possible selection as specified in $\SWC$.
    Analogously we define $B'_i$, $C'_i$, and $T'_i$ for $\zeta'_i$.

    We construct a solution SOL and inductively show that the value of each considered cell $(\zeta_i,\zeta'_j)$ and $(\zeta'_i,\zeta_j)$ is at most a factor $(1+O(\varepsilon))$ larger than the number of errors of an optimal solution restricted to the considered prefix and the considered rows. 
    Afterwards we show that our algorithm computes a solution at least as good as SOL.

    We first consider the DP cell $(\zeta_0,\zeta'_0)$.
    Recall that we assumed \WLOG that $i'_0 > i_0$.
    We apply Lemma~\ref{lem:swc-gap} with the parameters of the pair of cells to obtain the prefixes $\sigma_{\zeta_0}$ and $\sigma'_{\zeta'_0}$.
    The total number of errors within the columns of $M$ at the prefixes is therefore at most a factor $(1+O(\varepsilon))$ larger than in $(\tau,\tau')$.
    There are two possibilities for the subsequent steps with $i \ge 0$.

    We first assume that $b_{i+1}' > b_{i}$ and consider the cell $(\zeta_{i},\zeta'_{i+1})$.
    Then, similar to the proof of Lemma~\ref{lem:simpleDP}, we apply $\SWC$ to obtain the suffix of $({\sigma}_{\zeta_i,\zeta'_{i+1}},{\sigma'}_{\zeta_i,\zeta'_{i+1}})$ after $\en_{B_{i'}}$.
    By Lemma~\ref{lem:swc-gap}, considering the suffix alone we have at most a factor $(1+O(\varepsilon))$ more errors within these columns than $(\tau,\tau')$.

    Since $\zeta_i$ is a predecessor of $\zeta_{i+1}$, all newly assigned rows were not considered in $(\zeta_i,\zeta'_i)$.
    Note that $\zeta_i$ did not change. Even though we looked at the same chunks, we used the same selections and therefore did not change $\sigma_{\zeta_i,\zeta'_i}$.

    The second possibility is that $b_{i+1}' < b_{i}$ and we consider the cell $(\zeta_{i+1},\zeta'_{i})$. The instance is shown in Fig.~\ref{fig:DP-crux2}.
    We then apply $\SWC$ to obtain the suffix of $({\sigma}_{\zeta_{i+1},\zeta'_{i}},{\sigma'}_{\zeta_{i+1},\zeta'_{i}})$ after $\en_{B_{i}}$.
    We obtain a $(1+O(\varepsilon))$-approximation analogous to the case $b_{i+1}' > b_{i}$.
\end{proof}

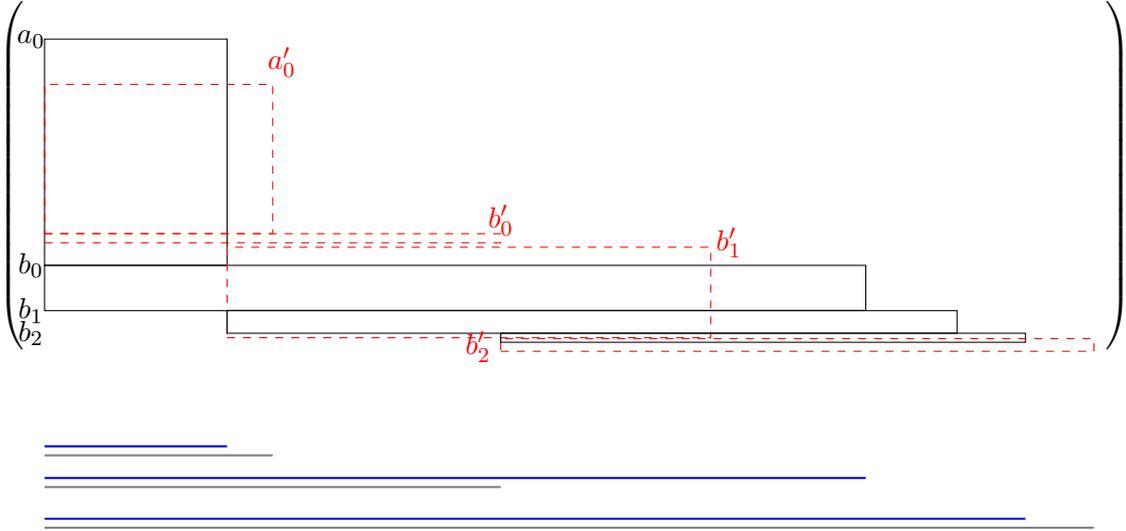
\begin{figure}[h]
    \begin{center}
        \begin{tikzpicture}[scale=0.6]
            \draw (-2,-2) rectangle (2,3);      
            \draw[red, dashed] (-2,2) rectangle (3,-1.3);
            \node at (-2.3,3) {$a_0$};
            \draw (-2,-3) rectangle (16,-2);     
            \node at (-2.3,-2) {$b_0$};
            \node at (-2.3,-3) {$b_1$};
            \node at (-2.3,-3.5) {$b_2$};
            \draw[red,dashed] (-2,-1.3) rectangle (8,-1.5);
            \draw (2,-3.5) rectangle (18,-3);
            \draw[red,dashed] (2,-1.6) rectangle (12.6,-3.6);
            \draw[red,dashed] (8,-3.62) rectangle (21,-3.9);
            \draw (8,-3.7) rectangle (19.5,-3.5); 
            \node[red] at (3.2,2.5){$a'_0$};
            \node[red] at (8,-1) {$b_0'$};
            \node[red] at (13,-1.5) {$b'_1$};
            \node at (-2.7,0){\(\left(\rule{0cm}{2.4cm}\right.\)};
            \node at (21.5,0){\(\left.\rule{0cm}{2.4cm}\right)\)};
            \draw (2,-3) to (2,-3.5);
            \node[red] at (7.5,-3.8) {$b'_2$};
            \draw[blue, thick] (-2,-6) -- (2,-6);
            \draw[gray, thick] (-2,-6.2) -- (3,-6.2);
            \draw[blue, thick] (-2,-6.7) -- (16,-6.7);
            \draw[gray, thick] (-2,-6.9) -- (8,-6.9);
            \draw[blue, thick] (-2,-7.6) -- (19.5,-7.6);
            \draw[gray, thick] (-2,-7.8) -- (21,-7.8);
        \end{tikzpicture}
        \caption{\label{fig:DP-crux2} Blocks of an instance $M$ in the DP for a pair of solution strings. The blue and gray lines represent $\sigma$ and $\sigma'$ respectively from first two iterations of DP. The sketch shows the switch example in the second iteration because $b_{1}' < b_{0}$.}
    \end{center}
\end{figure}

\section{Subinterval-free instances.}\label{sec:subinterval-free}
We show how to generalize the results of the previous section in order to handle instances where no interval of a string $s$ is a proper subinterval of a string $s'$ and thus show Theorem~\ref{thm:ptas}.
To this end, we first show how to handle the rooted version of sub-interval free instances, where there is one column $j$ such that each string of the instance crosses $j$.

We order the rows of a subinterval-free instance $M$ from top to bottom such that for each pair $i,i'$ of rows with the binary part of $i$ starting on the left of the binary part of $i'$, $i$ is above $i'$.
In other words, the binary strings are ordered from top to bottom with increasing starting position (\ie, column).
Observe that the sub-string freeness property ensures that the last binary entry of $i'$ is not on the left of the last binary entry of $i$.

\begin{lemma}\label{lem:second_instance_setting}
    Let $M$ be a $\GMEC$ instance such that no string is the substring of another string.
    Furthermore we assume that there is a column $j$ of $M$ such that each string of the instance crosses $j$.
    Then there is a PTAS for $M$.
\end{lemma}

\begin{proof}
    Let $s$ and $t$ be the first and the last row of $M$.
    The column $j$ determines a block $W$ of $M$ that spans all rows and the columns from the first binary entry of $t$, $j_t$, to the last binary entry of $s$, $j_s$.
    In particular, $W$ has only binary entries.

    The right hand side of $j_t$ (the submatrix of $M$ composed of all columns with index at least $j_t$) forms a $\GMEC$ instance as required in Theorem~\ref{thm:column1}.
    The submatrix of $M$ that contains all rows of $M$ and columns $1$ to $j_s$ forms a \GMEC instance as required in Theorem~\ref{thm:column1} if we invert both the order of the rows and the columns.
    Instead of changing the ordering of the matrix, we can run the algorithm from right to left and from bottom to top.

    We would like to apply Theorem~\ref{thm:column1} independently to the two specified sub-problems.
    To this end we define a special set of DP cells $\gamma$ with cells $(\zeta_W,\zeta'_W) \in \gamma$.
    The content of these cells is similar to the regular cells, but it contains the information for both sides simultaneously.
    More precisely, a cell $\zeta_W$ has the following entrees (see also Figures~\ref{fig:second_instance1} and \ref{fig:second_instance2}).

    (a) Three consecutive ranges of rows determined by numbers $1\le \overleftarrow{c} < \overleftarrow{b} < \overrightarrow{b} < \overrightarrow{c} \le n$.
    These numbers determine an upper range $R_U$ from row $\overleftarrow{b}+1$ to row $\overrightarrow{b}-1$ and the following further ranges.
    A left lower range $\overleftarrow{R}_L$ from row $\overleftarrow{b}$ to row $\overleftarrow{c}+1$,
    as well as a right lower range $\overrightarrow{R}_L$ from row $\overrightarrow{b}$ to row $\overrightarrow{c}-1$.
    (b) A separation into chunks $C$. There are $3/\varepsilon^2$ chunks in $C$: $1/\varepsilon^2$ for $R_U$,  $1/\varepsilon^2$ for $\overrightarrow{R}_L$, and  $1/\varepsilon^2$ for $\overleftarrow{R}_L$.
    (c) A selection $T$ of $3/\varepsilon^5$ rows (with repetition): $1/\varepsilon^2$ for each chunk.

    We analogously obtain $\zeta'_W$ with the same variables but marked with the symbol prime.
    The rows selected in $\zeta'_W$ are required to be disjoint from those in $\zeta_W$, i.e., $T \cap T' = \emptyset$.
    Also the boundaries of chunks in $\zeta_W$ and $\zeta'_W$ have to be disjoint.

    \begin{definition}[Center cells]
        The cells $(\zeta_W,\zeta'_W) \in \gamma$ are called \emph{center cells}.
        \label{def:center-cells}
    \end{definition}
    The reason is that they take a special role as common ``centers'' of two separate runs of the DP: one run to the left and one run to the right.
    Observe that for each feasible entry of $(\zeta_W,\zeta'_W)$, we can apply Theorem~\ref{thm:column1} independently to the left and to the right, since the DP cells $(\zeta_W,\zeta'_W)$ takes the role of the left-most cell in Theorem~\ref{thm:column1}.
    The strings only overlap between the columns $j_t,j_s$ where we obtain an instance of \BMEC, which in particular is a good SWC-instance.
    Note that for each column $\hat{j}$ on the right hand side of $j_t$, all rows of $W$ located above $\overleftarrow{b}$ with binary entry at column $\hat{j}$  have also a binary entry at all rows between $\overleftarrow{b}$ and $\overrightarrow{b}$, due to the subinterval-freeness. 
    The properties of $\hat{j}$ on the left hand side of $j_s$ are analogous.
    We will choose $\overleftarrow{b}$ and $\overrightarrow{b}$ in such a way that by Lemma~\ref{lem:swc-gap}, it is therefore sufficient to consider the rows between $\overleftarrow{b}$ and $\overrightarrow{b}$ in order to handle all rows crossing $j$.

    None of the remaining steps from Section~\ref{sec:second} interfere with each other.
    We therefore run the following $DP$.
    We first compute all center cells $(\zeta_W,\zeta'_W) \in \gamma$.
    For each cell, we store an infix of $\sigma$ and an infix of $\sigma'$.
    The infix of $\sigma$ starts at $j_t$ and ends at $j_s$.
    The entries of the two strings are those that we obtain from $\SWC$ with the parameters of $(\zeta_W,\zeta'_W)$.
    Each cell $(\zeta_W,\zeta'_W)$ forms a starting point for Algorithm $\textsc{SWC}^{\sigma, \sigma'}$, applied independently towards the left hand side and the right hand side.

    To see that the DP yields a good enough approximation, again we compare against an optimal solution $(\tau,\tau')$.
    Clearly we get a $(1+O(\varepsilon))$-approximation for the infix between column $j_t$ and $j_s$ if for a DP cell $(\zeta_W,\zeta'_W)$, by Lemma~\ref{lem:swc-gap}.
    Note that the computed solution does not consider the rows above $\overleftarrow{c}$ or below $\overrightarrow{c}$.
    Since the further processing respects our choice between $\overleftarrow{c}$ and $\overrightarrow{c}$, the claim follows from Theorem~\ref{thm:column1}. 
\end{proof}

\subparagraph{General sub-interval-free instances} 
We use Lemma~\ref{lem:second_instance_setting} to handle general sub-interval free instances.
Instead of a single column $j$ crossed by all strings, we determine a sequence $q = (q_1,q_2,\dotsc)$ of columns with the property that each string crosses exactly one of them.
Let $s_1$ be the first string in $M$.
Then we choose $q_1$ to be the column of the last entry of $s_1$.

We recursively specify the remaining columns.
For a given $j$ such that we know $q_j$, let $s_i$ be the last (\ie, bottom-most) string that crosses $q_j$.
Then we choose $q_{j+1}$ to be the last (\ie, rightmost) column of string $s_{i+1}$.
For each $q_i$ in the sequence $q$, we determine a block $W_i$ analogous to $W$ in Lemma~\ref{lem:second_instance_setting}.

A simple induction shows that by the no-substring property and the chosen order of strings, each string crosses at least one column of $q$ and none of them crosses more than one.
In particular, for each $j$, the solution on the left hand side of $q_j$ depends on rows of $M$ disjoint from the rows that determine the solution on the right hand side of $q_{j+1}$. 

In order to combine the solution on the right hand side of $q_j$ with the solution on the left hand side of $q_{j+1}$, we introduce a notion of dominance.
Let us consider two arbitrary submatrices $V_1$ and $V_2$ of $M$.
\begin{definition}[Dominance]
    \label{def:dominance}
    We say that $V_1$ $\tau$-dominates $V_2$ if for each column $c$ that is in both $V_1$ and $V_2$, either at least one of the two matrices has no binary entries or
    the number of binary entries in $\tau(V_1)$ is at least $1/\varepsilon^2$ times the number in $\tau(V_2)$.
    We say that $V_1$ is $\tau$-dominant over $V_2$ for a column $c$, if the one column submatrix of $V_1$ determined by $c$ dominates $V_2$. 
    
    We analogously define $\tau'$-dominance.
\end{definition}

Consider a submatrix $\overrightarrow{V}$ of $M$ that only contains rows that cross $q_i$ and a submatrix $\overleftarrow{V}$ of $M$ that only contains rows that cross $q_{i+1}$.
We observe that if $\overrightarrow{V}$ is $\tau$-dominant over $\overleftarrow{V}$ for some column $c$, it is also $\tau$-dominant for all columns on the left hand side of $c$:
until $q_i$ is reached, when moving to the left the number of binary entries of $\tau(\overrightarrow{V})$ increases and the number of binary entries of $\tau(\overleftarrow{V})$ decreases.
Analogously, if $\overleftarrow{V}$ is $\tau$-dominant over $\overrightarrow{V}$ for some column $c$, it is also $\tau$-dominant for all columns on the right hand side of $c$.

We therefore have a possibly empty interval $I$ without $\tau$-dominance such that the columns of $\overrightarrow{V}$ on the left hand side of $I$ are $\tau$-dominant and the columns of $\overleftarrow{V}$ on the right hand side of $I$ are $\tau$-dominant.
(See also Figures~\ref{fig:second_instance1} and \ref{fig:second_instance2}.)

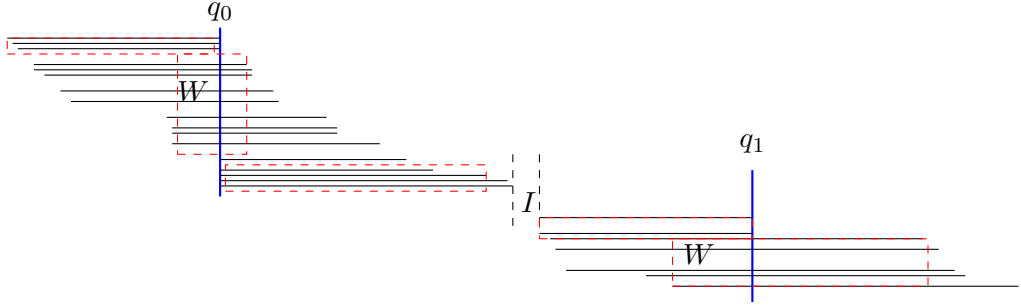
\begin{figure}[h]
    \begin{center}
        \begin{tikzpicture}[scale=0.7]
            \draw (0,0) -- (4,0);
            \draw (0.1,-.1) -- (4,-.1);
            \draw (0.2,-.2) -- (4,-.2);
            \draw[blue,thick] (4,.2) -- (4,-3);
            \node at (4,.5) {$q_0$};
            \node at (3.5,-1) {$W$};
            \draw (0.5,-.5) -- (4.5,-0.5);
            \draw (0.5,-.6) -- (4.6,-0.6);
            \draw (0.7,-.7) -- (4.6,-0.7);
            \draw (1,-1) -- (5,-1);
            \draw (1.2,-1.2) -- (5.1,-1.2);
            \draw (3,-1.5) -- (6,-1.5);
            \draw (3.1,-1.7) -- (6.2,-1.7);
            \draw (3.1,-1.8) -- (6.2,-1.8);
            \draw (3.1,-2) -- (7,-2);
            \draw (4,-2.3) -- (7.5,-2.3);
            \draw (4,-2.5) -- (8,-2.5);
            \draw (4,-2.6) -- (9,-2.6);
            \draw (4,-2.7) -- (9.4,-2.7);
            \draw (4,-2.8) -- (9.5,-2.8);
            \draw[blue,thick] (14,-2.5) -- (14,-5);
            \node at (14, -2) {$q_1$};
            \draw (10,-3.4) -- (14,-3.4);
            \draw[dashed] (9.5,-2.2) to (9.5,-3.6);
            \draw[dashed] (10,-2.2) to (10,-3.6);
            \draw (10,-3.7) -- (14,-3.7);
            \draw (10.2,-3.8) -- (17.2,-3.8);
            \draw (10.3,-4) -- (17.5,-4);
            \draw (10.5,-4.4) -- (17.8,-4.4);
            \draw (12,-4.5) -- (18,-4.5);
            \draw (12.5,-4.7) -- (19,-4.7);
            \node at (9.8,-3.1) {$I$};
            \node at (13,-4.1) {$W$};
            \draw[red, dashed] (0,0) rectangle (3.89,-.3);
            \draw[red, dashed] (3.2,-.3) rectangle (4.5,-2.2);
            \draw[red, dashed] (10,-3.4) rectangle (14,-3.8);
            \draw[red, dashed] (12.5,-3.8) rectangle (17.3,-4.7);
            \draw[red, dashed] (4.1,-2.4) rectangle (9,-2.9);
        \end{tikzpicture}
        \caption{\label{fig:second_instance1} Blocks represented by ranges shown in red on an instance $M$ and the blue lines are the columns, $I$ and $W$ shows the empty interval and central region respectively.}
    \end{center}
\end{figure}

\begin{figure}[h]
    \begin{center}
        \begin{tikzpicture}[scale=0.7]
            \draw (0,0) -- (4,0);
            \draw (0.1,-.1) -- (4,-.1);
            \draw (0.2,-.2) -- (4,-.2);
            \draw[blue,thick] (4,.2) -- (4,-3);
            \node at (4,.5) {$q_0$};
            \node at (3.5,-1) {$W$};
            \draw (0.5,-.5) -- (4.5,-0.5);
            \draw (0.5,-.6) -- (4.6,-0.6);
            \draw (0.7,-.7) -- (4.6,-0.7);
            \draw (1,-1) -- (5,-1);
            \draw (1.2,-1.2) -- (5.1,-1.2);
            \draw (3,-1.5) -- (6,-1.5);
            \draw (3.1,-1.7) -- (6.2,-1.7);
            \draw (3.1,-1.8) -- (6.2,-1.8);
            \draw (3.1,-2) -- (7,-2);
            \draw (4,-2.3) -- (7.5,-2.3);
            \draw (4,-2.5) -- (8,-2.5);
            \draw (4,-2.6) -- (9,-2.6);
            \draw (4,-2.7) -- (9.4,-2.7);
            \draw (4,-2.8) -- (9.5,-2.8);
            \draw[blue,thick] (11,-2.5) -- (11,-5);
            \node at (11, -2) {$q_1$};
            \draw (5,-3) -- (11,-3);
            \draw (5.1,-3.2) -- (11.2,-3.2);
            \draw (5.2,-3.3) -- (11.2,-3.3);
            \draw (5.5,-3.4) -- (11.3,-3.4);
            \draw[dashed] (5.5,-2.2) to (5.5,-3.6);
            \draw[dashed] (9,-2.2) to (9,-3.6);
            \draw (8,-3.7) -- (12,-3.7);
            \draw (8.2,-3.8) -- (12.2,-3.8);
            \draw (8.3,-4) -- (12.5,-4);
            \draw (8.5,-4.4) -- (12.8,-4.4);
            \draw (10,-4.5) -- (13,-4.5);
            \draw (10.5,-4.7) -- (14,-4.7);
            \node at (7,-3.1) {$I$};
            \node at (10.8,-4.1) {$W$};
            \draw[red, dashed] (0,0) rectangle (3.89,-.3);
            \draw[red, dashed] (3.2,-.3) rectangle (4.5,-2.2);
            \draw[red, dashed] (10.5,-3.1) rectangle (11.2,-4.7);
        \end{tikzpicture}
        \caption{\label{fig:second_instance2} This sketch shows a non-dominance example in region $I$.}
    \end{center}
\end{figure}
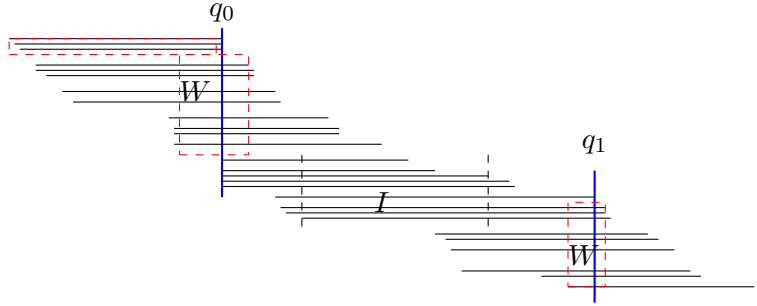

\begin{definition}[Dominance region]
    The \emph{dominance region} of $\overrightarrow{V}$  with respect to $\overleftarrow{V}$ is the set of columns where $\overrightarrow{V}$ is dominant over $\overleftarrow{V}$, and vice versa. 
    \label{def:dominance-region}
\end{definition}
Within the dominance region, our old DP can simply compute solutions without considering interferences: the dominated set of rows is small enough to be ignored, applying Lemma~\ref{lem:swc-gap}.

Within the interval $I$, the DP cells on both sides of $I$ have to ``cooperate.''
We obtain a \BMEC block in the middle with additional rows on the top and bottom. 
This sub-instance can be solved directly.

We use DP cells similar to Lemma~\ref{lem:second_instance_setting}, but for more than one center.
For each $j$, we consider column $q_j\in q$ and a collection $\kappa_j$ of DP cells $(\zeta_{q_j},\zeta'_{q_j}) \in \kappa_j$.
Each cell $(\zeta_{q_j},\zeta'_{q_j})$ is a center cell with center $q_j$.
We refer to the cells in $\kappa_j$ as the $j$th center cells.

Additionally, for each center cell we also store the dominance information on the left and right of $q_j$, 
i.e., we store the intervals $\overleftarrow{I}, \overleftarrow{I}'$ between $q_{j-1}$ and $q_j$ and the intervals $\overrightarrow{I},\overrightarrow{I}'$ between $q_j$ and $q_{j+1}$ where no cell dominates another, once with respect to $\tau$ and  once with respect to $\tau'$.

Formally this means to extend the cells by four numbers that store the start and end points four intervals $\overleftarrow{I},\overrightarrow{I}',\overrightarrow{I}$, and $\overrightarrow{I}'$.

For each of the four intervals we store additional information.
The four intervals only differ in whether we consider $\sigma$ or $\sigma'$. The left and right version are symmetric.
Therefore it is sufficient to analyze the details for a generic $I \in \overleftarrow{I},\overrightarrow{I}',\overrightarrow{I},\overrightarrow{I}'$.
The interval ${I}$ determines a block ${B}$ that we subdivide into chunks ${C}$ and we select rows ${T}$.
There are several differences to previous trisections, subdivisions and selections.

We divide the rows of block ${B}$ into four regions: 
a middle part ${U^\uparrow}$ that has only binary entries (a \BMEC sub-instance) such that each row crosses $q_j$,
a middle part ${U^\downarrow}$ that has only binary entries such that each row crosses $q_{j+1}$,
the rows ${U}^{\uparrow\uparrow}$ above ${U}^\uparrow$, and the rows ${U}^{\downarrow\downarrow}$ below ${U}^\downarrow$.
We choose the two middle parts such that the number of rows is maximal.

It is not sufficient to use a globally guessed $r$. 
Instead, we add four numbers ${r}^{\uparrow\uparrow},{r}^\uparrow,{r}^\downarrow,{r}^{\downarrow\downarrow}$ to the DP cell in order to guess and store the values 
$\tau({U}^\uparrow)$, $\tau({U^\downarrow})$, $\tau({U^{\uparrow\uparrow}})$, and $\tau({U^{\downarrow\downarrow}})$.

Due to the non-dominance, we know that for each column of ${B}$, at least an $\varepsilon^2$-fraction of rows from $\tau(B)$ are located in the middle part ${M}$.
Observe that there is no region that takes the role of $X$ in a trisection.
We obtain an instance similar to a good SWC-instance, but it has two non-binary regions and the binary region only has an $\varepsilon^2$ fraction of rows instead of an $\varepsilon$-fraction.
To be able to still apply Lemma~\ref{lem:swc-gap}, we subdivide each of the three regions into chunks and increase the number of chunks per region. 
The number of chunks depends on  the four versions of ${r}$. 
If our chunks do not contain more than $\varepsilon^4 {r}$ rows of $\tau(M)$, the lemma is applicable.
We guess a number $k$ and set the size of chunks with root $q_j$ to contain $\varepsilon^{4k} {r}^{\uparrow\uparrow}$ rows of $\tau(M)$ in ${U}^\uparrow$ and $\varepsilon^{4k} {r}^{\uparrow\uparrow}$ in ${M}$ for the rows with root $q_j$.
There may be an additional chunk with fewer rows, if the numbers don't match.
We specify the remaining chunks symmetrically, based on a number guessed for root $q_{j+1}$.
The choice of $k$ will become clear in the description of the DP.

The increased precision also requires that we  increase the precision of the entire remaining DP: we replace each selection of $1/\varepsilon^2$ chunks into selections of $1/\varepsilon^4$ chunks.
Clearly, the increased precision cannot decrease the quality of the computed solution.

The idea of the DP is that for each $q_j$, we run the rooted DP as an \emph{inner} DP that determines solutions for their dominance regions that fit to solutions in the consecutive non-dominance regions.
The non-dominance regions then form interfaces that we can use to compute an overall solution from left to right with an \emph{outer} DP.

The inner DP works as follows.
For each cell $(\zeta_{q_1},\zeta'_{q_1}) \in \kappa_1$, we compute the prefixes of $\sigma,\sigma'$ until $q_1$ exactly as in Lemma~\ref{lem:second_instance_setting}.
We start the DP to the right hand side also the same way as before, but with the difference that as soon as we reach the row ranges for $\overrightarrow{I}$ or $\overrightarrow{I}'$, we use the choices already stored in $(\zeta_{q_1},\zeta'{q_1})$.
We have to ensure that our choices within the DP do not contradict the choices of  $(\zeta_{q_1},\zeta'_{q_1}) \in \kappa_1$. 
If $(\zeta_{q_1,i},\zeta'_{q_1,i})$ is the cell of the inner DP that overlaps with $\overrightarrow{I}$ first, we require that among the common rows, $\overrightarrow{B}$ contains the remaining rows from $\tau(M)$ restricted to the rows of the inner DP and does not contradict $B_i$. 
Each chunk of $C_i$ contains $\varepsilon^k |\tau(M)|$ rows of $\tau(M)$, for some integer $k$ (the same $k$ that we guessed for the non-domination region). 
We have to ensure that the chunks of $\overrightarrow{U}^\uparrow$ and $\overrightarrow{M}^\uparrow$ match the chunks and the chunks of $\overrightarrow{C}_i$.
Furthermore, we have to check that the selection of rows matches.

For all $j > 1$ continue in the same manner starting from $(\zeta_{q_j},\zeta'_{q_j})$ and handle the processing of $\overleftarrow{I}$ as we did before with $\overrightarrow{I}$.
Observe that we can see the processing of $(\zeta_{q_1},\zeta'_{q_1})$ as a special case with empty interval $\overleftarrow{I}$, and to obtain the suffix of $\sigma,\sigma'$, the last interval $\overrightarrow{I}$ can be handled as empty interval.

The global DP proceeds from left to right. For each $q_j$, it considers all cells $(\zeta_{q_j},\zeta'_{q_j})$.
The value of $(\zeta_{q_j},\zeta'_{q_j})$ is its inner DP value plus the best value achievable on the left hand side with the same choice of parameters for the left non-domination region.
Among all cells from $\kappa_j$ with the same parameters for the right non-domination region, the global DP only keeps the best value (the smallest number of errors).

\subparagraph{The above DP is a PTAS for $\mathbf{M}$.} 
Let $(\tau,\tau')$ be an optimal solution.
For each separate $q_j \in q$, we run the same DP as in Lemma~\ref{lem:second_instance_setting} and thus we obtain a $(1+O(\varepsilon))$-approximation.
For the intervals $\overleftarrow{I}$ and $\overrightarrow{I}$, there is a choice of parameters that matches the choices analyzed in Lemma~\ref{lem:second_instance_setting}.
We therefore only have to argue that the transition between sub-instances works correctly.
We consider the dominant regions determined by $(\tau,\tau')$ and consider the DP cells that guess these regions correctly from left to right.
Let $I$ be one of the guessed non-dominant regions.
We obtain the solution for $I$ by applying $\SWC$, which gives a $(1+O(\varepsilon))$ approximation.
The transition between dominant and non-dominant regions uses that in both cases we create the solution strings from the same parameters in $\SWC$ and therefore creates the solution from the same instance strings.
This finishes the proof of Theorem~\ref{thm:ptas}.

\section{A QPTAS for general instances.}\label{sec:QPTAS}

To solve the general instances, the main observation is that we divide the rows into their at most $\log_2(m)$ length classes $\Lambda_i$, and the $i$th length class $\Lambda_i$ is the set of all strings of length $\ell$ with $\ell \in (m/2^{i+1}, m/2^{i}]$.
First we present an algorithm to solve each length class $\Lambda_i$ separately by constructing their corresponding columns.

\subsection{Length classes.}\label{sec:length-classes}
We show how we can handle length classes of strings.
To this end, let us assume \WLOG that $m$ (\ie, the number of columns in $M$) is a power of $2$.
Then for each $i \ge 0$, the $i$th length class $\Lambda_i$ is the set of all strings of length $\ell$ with $\ell \in (m/2^{i+1}, m/2^{i}]$. 
We observe the following known property of length classes.
\begin{lemma}\label{lem:half}
    For each $i \ge 0$ there is a set $q_i = \{q_{i,1},q_{i,2},\dotsc\}$ of columns such that (a) each string in $\Lambda_i$ crosses at least one column from $q_i$ and (b) no string from $\Lambda_i$ crosses more than two columns from $q_i$.
    Furthermore, we can choose the sets such that $q_i \subseteq q_{i+1}$.
\end{lemma}
\begin{proof}
    At level $i$, for each $k$ with $1 \le k \le 2^{i+1}$ we select the column with index $k \cdot m/2^{i+1}$.
    We observe that the distance between two consecutive columns from $q_i$ is $m/2^{i+1}$, which matches the shortest length of strings in $\Lambda_i$: 
    if a minimal string starts right after a column of $q_i$, its last entry will cross the next column of $q_i$. 

    Since strings do not start before column $1$ and column $m$ is contained in each $q_i$, claim (a) follows.
    To see (b), observe that a maximum length string of $\Lambda_i$ is at most $m/2^i$.
    Let $j$ be an index. 
    The number of columns from $q_{i,j}$ to the column right before $q_{i,j+1}$ and from $q_{i,j+1}$ to right before $q_{i+2}$ are exactly $m/2^{i+1}$ .
    If the string starts directly at a column $q_{i,j}$ from $q_i$, it would cross column $q_{i,j+1}$ and end right before column $q_{i,j+2}$.
    
    The last claimed property follows directly from the construction of the sets $q_i$. 
    (See also Fig.~\ref{fig:gen_instance1}).
\end{proof}

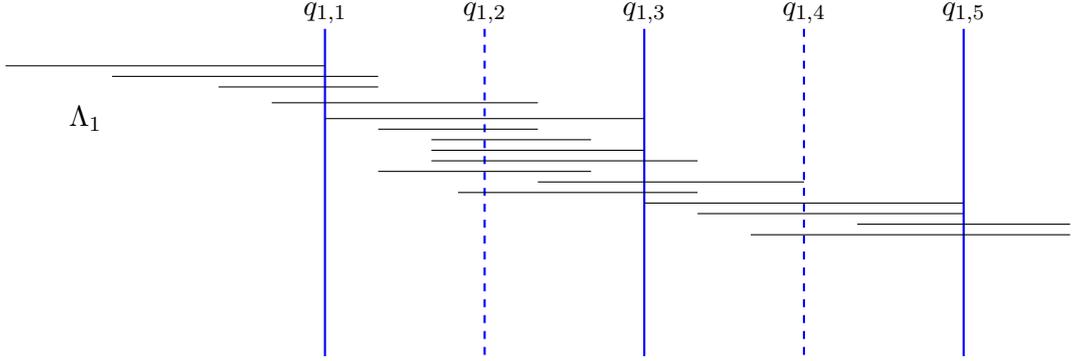
\begin{figure}[h]
    \begin{center}
        \begin{tikzpicture}[scale=0.7]
            \draw[blue,thick] (7,.2) -- (7,-6);
            \draw[blue,thick] (13,.2) -- (13,-6);
            \draw[blue,thick] (19,.2) -- (19,-6);
            \draw[blue,thick,dashed] (10,.2) -- (10,-6);
            \draw[blue,thick,dashed] (16,.2) -- (16,-6);
            \node at (2.5,-1.5) {$\Lambda_1$};
            \node at (7,0.5) {$q_{1,1}$};
            \node at (10,0.5) {$q_{1,2}$};
            \node at (13,0.5) {$q_{1,3}$};
            \node at (16,0.5) {$q_{1,4}$};
            \node at (19,0.5) {$q_{1,5}$};
            \draw (1,-0.5) -- (7,-0.5);
            \draw (3,-0.7) -- (8,-0.7);
            \draw (5,-0.9) -- (8,-0.9);
            \draw (6,-1.2) -- (11,-1.2);
            \draw (7,-1.5) -- (13,-1.5);
            \draw (8,-1.7) -- (11,-1.7);
            \draw (9,-1.9) -- (12,-1.9);
            \draw (9,-2.1) -- (13,-2.1);
            \draw (9,-2.3) -- (14,-2.3);
            \draw (8,-2.5) -- (12,-2.5);
            \draw (9.5,-2.9) -- (14,-2.9);
            \draw (11,-2.7) -- (16,-2.7);
            \draw (13,-3.1) -- (19,-3.1);
            \draw (14,-3.3) -- (19,-3.3);
            \draw (17,-3.5) -- (21,-3.5);
            \draw (15,-3.7) -- (21,-3.7);
        \end{tikzpicture}
        \caption{\label{fig:gen_instance1} For a single-length-class instance, the sketch shows the strings crossing each column either exactly once or exactly twice.}
    \end{center}
\end{figure}

For each $i$, we now separate $\Lambda_i$ into two sub-instances.
One sub-instance $\Lambda'_i$ is formed by those rows from $\Lambda_i$ that only cross one column of $q_i$ and the second sub-instance $\Lambda''_i$ is formed by those rows that cross exactly two columns of $\Lambda_i$.

\begin{definition}[DP for a length class $\Lambda_i$]
    For each index $j$ let $\xi'_j$ be the sets of DP cells for $\Lambda'_i$ and for the odd indices $j$ let $\xi''_j$ be the set of cells for $\Lambda''_i$.
    We define a super-cell that starts in $j$, $(Z'_j,Z''_j,Z'_{j+1},Z''_{j+2}) \in \xi'_j \times \xi''_j \times \xi'_{j+1} \times \xi''_{j+2}$ and the super-cell that ends in $j$, $(Z'_{j-1},Z''_{j-2},Z'_j,Z''_j) \in \xi'_{j-1} \times \xi''_{j-2} \times \xi'_{j} \times \xi''_{j}$.
    \label{def:dp-whole-length-class}
\end{definition}

\begin{lemma}\label{lem:length-class}
    There is a QPTAS for \GMEC if all strings are in the same class $\Lambda_i$.
\end{lemma}
To prove it, we consider DP-cells according to Definition~\ref{def:dp-whole-length-class} and combine these cells from two consecutive columns such that they are compatible. 
\begin{proof}
    To combine the PTAS for $\Lambda'_i$ and $\Lambda''_i$, we proceed from left to right.
    For each index $j$ let $\xi'_j$ be the sets of DP cells for $\Lambda'_i$ and for the odd indices $j$ let $\xi''_j$ be the set of cells for $\Lambda''_i$.
    For each column $c$ before $q_{i,1}$, each pair of cells $(Z,Z') \in \xi'_1 \times \xi''_1$ determines two subproblems for which we compute the two separate solutions.
    Let $\mathcal{B}(Z,Z',z)$ be the set of all boxes (sets of rows) for $\sigma$ considered in the sub-cells of $(Z,Z')$ at column $z$ and
    let $\mathcal{B}'(Z,Z',z)$ be the set of all boxes (sets of rows) for $\sigma'$ considered in the sub-cells of $(Z,Z')$ at column $z$. 

    We now extend the DP as follows.
    We compose the solution from left to right, starting with the prefix of $(\sigma,\sigma')$ before $q_{i,1}$ and then, step by step, we fill the intervals between $q_{i,j}$ and $q_{i,j+1}$ for $j \ge 1$.
    The starting interval can be seen as the interval between a dummy-column $q_0$ and $q_1$.
    For each $j$, let us analyze its interval. 
    If $j$ is odd, we simultaneously consider the cells $\xi'_j,\xi'_{j+1},\xi''_j, \xi''_{j+2}$.
    Otherwise, we simultaneously consider the cells $\xi'_j,\xi'_{j+1},\xi''_{j-1}, \xi''_{j+1}$.

    For each column $c$ with index $\ell$ in the interval, in both cases the values of the DP cells reveal all $\SWC$ instances at $c$ that we would have to solve in order to obtain solutions for $\Lambda'_i$ and $\Lambda''_i$ separately.
    Instead of solving these instances separately, we solve them simultaneously.

    Let $\hat{C}_1, \hat{C}'_1$ and $\hat{C}_2,\hat{C}'_2$ be the chunks of the four DP cells at position $j$.
    In order to determine the value $\sigma_j$, we have to combine $\hat{C}_1$ with $\hat{C}_2$ and take care of the different densities of rows.
    To this end, we generalize the function $\textsc{Majority}_j$.
    \begin{definition}{Generalized Majority}
        \label{def:generalized-majority}
        For a single chunk $c$, let $r(c)$ be the number of rows in $c$ that are guessed to be in $\tau(M)$ and $t(c)$ the number of selected rows.
        As in Definition~\ref{def:weighted-majority}, we replace all values zero by $-1$.
        Then, for the given set of chunks $C$ with selection $T$, we compute 
        \[
            \rho := \sum_{c \in C} \sum_{i \in T\colon i \in c} (r(c) \cdot M_{i,j})\,.
        \]
        We set $\sigma_j = 1$ if the outcome is at least zero and $0$ otherwise.    
        The definition is analogous for $\sigma'_j$.
    \end{definition}

    By replacing the majority function by the generalized majority function of Definition~\ref{def:generalized-majority} in the proof of Lemma~\ref{lem:SWC}, we obtain a $(1 + O(\varepsilon))$-approximation
    also if we consider different cells simultaneously.

    Since we consider all cells for the entire interval simultaneously, one of the choices is at least as good as sampling uniformly at random with knowledge of $\tau(M)$ and $\tau'(M)$.
    We therefore obtain a solution for the interval with at most a $(1+O(\varepsilon))$ factor of errors compared to $(\tau,\tau')$

    Finally we have to join the results that we obtain for the intervals.
    Observe that for each pair of cells $(Z''_j,Z''_{j+2}) \in \xi''_j \times \xi''_{j+2}$ there are two consecutive pairs of cells
    $(Z'_j,Z'_{j+1}) \in \xi'_j \times \xi'_{j+1}$ and $(Z'_{j+1}, Z'_{j+2}) \in \xi'_{j+1} \times \xi'_{j+2}$.
    For a quadruple of cells $(Z'_j,Z''_j,Z'_{j+1},Z''_{j+2})$ we consider each quadruple on the left hand side ending with the matching cells $Z'_j,Z''_j$.
    Among these, we take the one with fewest errors.
    To obtain the value of the new quadruple, we add the errors in the interval $(q_{i,j} q_{i,j+1}]$ to the value of the selected predecessor quadruple.
    To compute the value $(Z'_{j+1},Z''_j,Z'_{j+2},Z''_{j+2})$, we consider all cells $(Z'_j,Z''_j,Z'_{j+1},Z''_{j+2})$, \ie, the cells that have the same $Z''_{j},Z''_{j+1},Z'_{j+1}$ for all choices of $Z'_j$.
    We add the errors between $q_{i,j+1}$ and $q_{i,j+2}$ to the smallest value found among the predecessors.

    The approximation ratio follows from Lemma~\ref{lem:half-length-class} and the quasi-polynomial running time from the fact that we only consider constantly many super-cells simultaneously.

\end{proof}

\begin{lemma}\label{lem:half-length-class}
    There is a QPTAS for \GMEC if all strings are in the same class $\Lambda'_i$ or $\Lambda''_i$.
\end{lemma}
\begin{proof}
    We first note that by skipping all $q_{i,j}$ with even $j$, the strings in $\Lambda''_i$ cross exactly one column of the set.
    It is therefore sufficient to handle $\Lambda'_i$.

    For each column $q_{i,j}$, we create a set of DP cells (as defined in Definition~\ref{def:dp-cell-two-solution-string}) that stores information about a center region as defined in Definition~\ref{def:center-cells} and about non-domination intervals as defined in Definition~\ref{def:dominance}, exactly as in the proof of Theorem~\ref{thm:ptas}. 
    The next insight is that we can order the rows crossing column $c$ at the left and right side as defined below.
    \subparagraph{Left row ordering}
    We first order the rows with increasing starting positions of strings as in Lemma~\ref{lem:second_instance_setting}.
    At the left side of column $c$ , we obtain a similar instance as in Lemma~\ref{lem:second_instance_setting}.

    \subparagraph{Right row ordering}
    Afterwards we reorder the rows in order to handle the right hand side of column $c$. More precisely, we order the strings in increasing order based on the end of strings $e_i$.
    The obtained structure corresponds to the right hand side of column $c$ is similar to the instance handled in Lemma~\ref{lem:second_instance_setting}.

    Instead of running the DP of Lemma~\ref{lem:second_instance_setting}, we guess the sequence of blocks.
    An optimal solution $(\tau,\tau')$ determines a sequence of blocks $\overleftarrow{A}_1,\overleftarrow{A}_2,\dotsc,\overleftarrow{A}_k$ such that $|\tau(\overleftarrow{A}_{i+1})| =  \varepsilon^2 |\tau(\overleftarrow{A}_{i})|$. 
    Instead of moving from $\overleftarrow{A}_1$ to $\overleftarrow{A}_k$ using a DP, we directly guess the strings for all $k$ sub-matrices \emph{simultaneously}.
    We do the same with the chunks and row selections.
    Additionally, we guess the sequence of sub-matrices $\overleftarrow{A}'_1,\overleftarrow{A}'_2,\dotsc,\overleftarrow{A}'_{k'}$ simultaneously 
    such that $|\tau'(\overleftarrow{A}'_{i+1})| =  \varepsilon^2 |\tau'(\overleftarrow{A}'_{i})|$.
    We obtain a combined DP cell $\overleftarrow{\zeta}$ for $k+k'$ sub-matrices.

    Again we form the sub-matrices $\overrightarrow{A}_1,\overrightarrow{A}_2,\dotsc,\overrightarrow{A}_{k}$ and  $\overrightarrow{A}'_1,\overrightarrow{A}'_2,\dotsc,\overrightarrow{A}'_{k'}$ analogous to the left hand side and guess the selected strings of all matrices simultaneously 
    such that $|\tau(\overrightarrow{A}_{i+1})| =  \varepsilon^2 |\tau(\overrightarrow{A}_{i})|$ and  $|\tau'(\overrightarrow{A}'_{i+1})| =  \varepsilon^2 |\tau'(\overrightarrow{A}'_{i})|$. 
    We obtain a combined DP cell $\overrightarrow{\zeta}$ for all $k+k'$ sub-matrices on the right hand side.
    \begin{definition}[DP cell for sub-class of a length class $\Lambda_i$]
        For each $q_{i,j}$, let $\xi_j$ be the set of super-cells $(\overleftarrow{\zeta}, \overleftarrow{\zeta'}, \overrightarrow{\zeta}, \overrightarrow{\zeta'})$, but with the additional center and non-domination information of Lemma~\ref{lem:second_instance_setting}.
        \label{def:dp-cell-length-class}
    \end{definition}

    For each $q_{i,j}$, let $\xi_j$ be the set of super-cells based on Definition~\ref{def:dp-cell-length-class}.
    We then design a DP that moves from left to right through the columns in $q_i$.
    The DP and its analysis now follow from the proof of Theorem~\ref{thm:ptas}, but we consider the left hand side and right hand side of each cell from $\xi_j$ simultaneously.

    To analyze the running time, we observe that $k$ and $k'$ are at most $O(\log_{1/\varepsilon}(n))$ since for each $i$ we assume that $|\tau(\overleftarrow{A}_{i+1})| =  \varepsilon |\tau(\overleftarrow{A}_{i})|$
    and $|\tau'(\overleftarrow{A}'_{i+1})| =  \varepsilon |\tau'(\overleftarrow{A}'_{i})|$. 
    The number of instances $\overrightarrow{A}_i$ and $\overrightarrow{A}'_i$ are also at most $O(\log_{1/\varepsilon}(n))$ each, for the same reason.

    We thus obtain super-cells that are combined of logarithmically many sub-cells with polynomial complexity.
    We obtain an overall super-cell which is a quadruple $(\overleftarrow{\zeta},\overleftarrow{\zeta}',\overrightarrow{\zeta},\overrightarrow{\zeta}')$, and we have to distinguish $\bigl(n^{O(1)}\bigr)^{4 \log_{1/\varepsilon}(n)} = n^{O(\log n)}$
    different cells, which is quasi-polynomial\footnote{We assume that $n$ and $m$ are polynomially related. This is justified because there are $n \cdot m$ entries of $M$ and therefore measuring in $m$ instead of $n$ would also give a quasi-polynomial complexity.}.

    We now analyze the performance guarantee. 
    For each column $j$, we obtain the values $\sigma_j$ and $\sigma'_j$ in almost the same way as we do in Lemma~\ref{lem:second_instance_setting}, but with the difference that we require consistency with all other rows sampled.
    For an optimal solution $(\tau,\tau')$, it is sufficient to only consider choices of rows such that all rows selected for $\sigma$ are in $\tau(M)$ and all rows selected for $\sigma'$ are in $\tau'(M)$.
    Such a selection of rows ensures consistency. 
    Note that we could apply the proof of Lemma~\ref{lem:second_instance_setting}  from the root to the left hand side and to the right hand side independently,
    if we knew $\tau(M)$ and $\tau'(M)$, just by avoiding wrong assignments.
    The simultaneous selection of all relevant rows ensures that we consider at least one selection of rows that satisfies these strong conditions.
    This solution is a $(1+\varepsilon)$ approximation by the proof of Lemma~\ref{lem:second_instance_setting}, and
    our DP computes a solution of at least the same quality since we consider the overall number of errors with respect to all sampled rows.
\end{proof}

Combining the two sub-classes gives a QPTAS for an entire length class.

\subsection{The general QPTAS.}\label{sec:generalQPTAS}
Finally we combine our insights to an algorithm for general instances by combining different length classes.
(See also Fig.~\ref{fig:gen_instance2}.) 

\begin{figure}[h]
    \begin{center}
        \begin{tikzpicture}[scale=0.7]
            \node at (2.5, .5) {$q_{3,1}$};
            \node at (5, .5) {$q_{2,1}$};
            \node at (7.5, .5) {$q_{3,3}$};
            \node at (10, .5) {$q_{1,1}$};
            \draw[blue,thick] (10,.2) -- (10,-6);
            \draw[blue,thick, dashed] (5,.2) -- (5,-6);
            \draw[blue,thick, loosely dashed] (2.5,.2) -- (2.5,-6);
            \draw[blue,thick, loosely dashed] (7.5,.2) -- (7.5,-6);
            \draw (.2,-4) -- (10,-4);
            \draw (1,-4.2) -- (11,-4.2);
            \draw (3,-4.4) -- (10,-4.4);
            \draw (6,-4.6) -- (13,-4.6);
            \draw (7,-4.8) -- (14,-4.8);
            \draw (7,-5) -- (12,-5);
            \node at (0,-4.5) {$\Lambda_1$};

            \draw (.2,-2.5) -- (5,-2.5);
            \draw (1,-2.6) -- (5.5,-2.6);
            \draw (3,-2.8) -- (7,-2.8);
            \draw (3.5,-3.0) -- (6,-3.0);
            \draw (4,-3.2) -- (8,-3.2);
            \node at (0,-3) {$\Lambda_2$};

            \draw (6,-2.7) -- (10,-2.7);
            \draw (6.5,-2.9) -- (10.5,-2.9);
            \draw (7,-3.4) -- (10.2,-3.4);

            \draw (.2,-0.5) -- (2.5,-0.5);
            \draw (1,-0.7) -- (3,-0.7);
            \draw (1.1,-0.9) -- (2.8,-0.9);
            \draw (1.2,-1.1) -- (2.7,-1.1);
            \draw (1.6,-1.3) -- (3.3,-1.3);

            \draw (2.6,-1.5) -- (5.1,-1.5);
            \draw (3.5,-1.7) -- (5.7,-1.7);
            \draw (3.7,-1.9) -- (5.5,-1.9);
            \draw (2.6,-1.0) -- (5.2,-1.0);
            \draw (5.3,-1.0) -- (7.8,-1.0);
            \node at (0,-1) {$\Lambda_3$};
        \end{tikzpicture}
        \caption{\label{fig:gen_instance2} Different length classes, $\Lambda_1$ with corresponding column $q_{1,1}$, 
            $\Lambda_2$ with corresponding columns $q_{2,1}, q_{2,2} = q_{1,1}$, 
            and $\Lambda_3$ with corresponding columns $q_{3,1}, q_{3,2} = q_{2,1}, q_{3,3}, q_{3,4} = q_{1,1}$.}
    \end{center}
\end{figure}
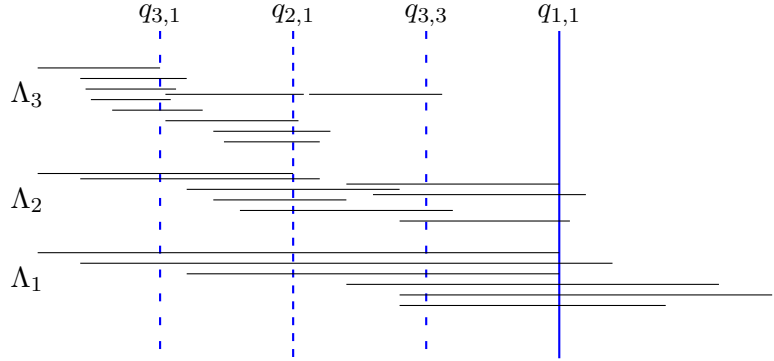

For different length classes $\Lambda_i$, we construct their corresponding columns as explained in the previous section. 
The main idea is that for each column $j$, we only have to consider those quadruple of super-cells according to Definition~\ref{def:dp-whole-length-class} that cross $j$ from all the length classes simultaneously.
We therefore consider at most $O(\log(n))$ quadruples of super-cells simultaneously. In the dynamic program, we consider a joint quadruple of super-cells from all the length classes.
Then the overall complexity of a joint cell is quasi-polynomial: the number of different cells is $\bigl(n^{O(\log n)}\bigr)^{O(\log n)} = n^{O(\log^2 n)}$.

Let $Q_{i,j}$ be the set of quadruples of length class $i$ crossing column $j$ such that the strings are ordered from shortest length class to the longest.
For each length class $i$, a quadruple $q \in Q_{i,j}$ is the set of rows starting at $j$, cross $j$, or end in $j$.
If $j$ is the index of $q_{i,\ell}$,
the quadruple $q$ starts in $j$ if it is formed by cells $(Z'_\ell,Z''_\ell,Z'_{\ell+1},Z''_{\ell+2})$ and ends in $j$ if it is formed by $(Z'_{\ell-1},Z''_{\ell-2},Z'_\ell,Z''_\ell)$ (see Definition~\ref{def:dp-whole-length-class}).
If $j$ lies between $q_{i,\ell}$ and $q_{i,\ell+1}$, $j$ crosses those quadruples that contain $Z'_\ell$ and $Z'_{\ell+1}$.
If non of the cases are true, we do not consider $q$ in the cells for column $j$.

Let us consider a $\log(n)$ vector of quadruples $v$, with one quadruple $Q_{i,j}$ for each $i$ and, consider quadruples starting at, ending at, or crossing column $j$ for length class $i$.
We require that if for some $i$, the quadruple $q \in Q_{i,j}$ ends at $j$, then for all the length classes $\Lambda_k$ with $k>i$ the same condition holds (with index larger than $i$).
This also implies that if for some $i$, the quadruple of length class $i$ starts at $j$, then the same also holds for all quadruples of shorter length classes (with index larger than $i$).
In particular, in order to be able to combine neighboring vectors of quadruples, we do not allow to mix starting and ending quadruples.
Let $\phi$ be the set of all $\log(n)$ vectors of tuples as described above (with one tuple of each length class).
The tuple for each length class is defined as in Lemma~\ref{lem:length-class} and the DP for general instances follows the ideas of Lemma~\ref{lem:length-class}:
We move from left to right column by column. In the initialization step, the joint DP cell is initialized based on Algorithm~\ref{alg:SWC} using $\phi$.
We guess the blocks, chunks and selections from each length class and consider them jointly in a DP cell.

For column $j$, let us consider a vector $v \in \phi$.
We distinguish whether $v$ has starting or ending quadruples. 
(One of the two cases must apply due to the shortest length class.)
For a $v \in \phi$ with starting quadruples,
let $d$ be the smallest number such that there is a quadruple of length class $d$ starting at $j$.
To compute $v$ we consider all $v' \in \phi$ with the following properties.
(a) $v'$ has the same quadruples for all length classes $d' < d$ and 
(b) for $d' \ge d$, the right hand sides of the quadruples of length class $d'$ in $v'$ compatible the left hand sides of the quadruples of $v$.
The super-cells from the left and right hand side are compatible if the intersecting strings from the left and right hand side are assigned to the same types of solution string $\sigma$ or $\sigma'$.

For a $v \in \phi$ with ending quadruples,
let $d$ be the smallest number such that there is a quadruple of length class $d$ ending at $j-1$. (In the very first column of the instance, we do not need this value.)
To compute $v$ we consider all $v' \in \phi$ with the following properties.
(a) $v'$ has the same quadruples for all length classes $d' < d$ and 
(b) for $d' \ge d$, the right hand sides of the quadruples of length class $d'$ in $v'$ match the left hand sides of the quadruples of $v$ in column $j-1$.
Then the value of $v$ is the sum of the minimum value over all such $v'$ and the number of errors in column $j$ obtained by applying $\SWC$ exactly as in the proof of Lemma~\ref{lem:length-class}. 

The approximation ratio follows by arguing that the expected number of errors at each column is at most $(1+O(\varepsilon))$ of OPT
(see Lemma~\ref{lem:length-class}).
This finishes the proof of Theorem~\ref{thm:qptas}.

\section*{Acknowledgment}
We would like to thank Tobias Marschall for helpful discussions.

\newpage
\appendix
\section*{Appendix}

\section{Proof of Lemma~\ref{lem:SWC}}\label{app:SWC}
\begin{proof}
    We show the claim by using a randomized argument. 
    To this end, we assume that for each $i$, the rows from $U_i$ and $L_i$ are selected uniformly at random from $U_i \cap \tau(M)$ and $L_i \cap \tau(M)$ and
    the rows from $U'_i$ and $L'_i$ are selected uniformly at random from $U'_i \cap \tau'(M)$ and $L'_i \cap \tau'(M)$.
    We argue that for each column, the expected number of errors is at most a factor $(1 + O(\varepsilon))$ larger than in an optimal solution.
    Then the claim follows from linearity of expectation and the fact that there is a selection with at most the expected number of errors. 

    We consider the $j$th column of $M$.
    Let $c := \tau(M)_{*,j}$, but without rows that have an entry ``$-$'' in column $j$.
    Let $p := |\{ i : c_i = 0\}|/|c|$ be the fraction of zeros in $c$.
    By swapping the zeros and ones we can assume \WLOG that $p \ge 1-p$, \ie, $p \ge 1/2$.
    Our assumption implies $\tau_j = 0$ and the optimal solution has $(1-p)|c|$ errors within $c$.

    The general idea of the proof is as follows.
    Suppose we would select exactly one row from $\tau(M)$ uniformly at random.
    Then with probability $p$, the algorithm has $(1-p)|c|$ errors in $c$ and with probability $(1-p)$ the number of errors is $p|c|$.
    Therefore the expected number of errors is $(p(1-p) + (1-p)p) |c| = 2 p (1-p) |c|$.
    We obtain the approximation ratio $2 p (1-p) |c| / ((1-p)|c|) = 2 p$.

    We will see that the approximation ratio improves with choosing several rows instead of a single one.
    Additionally, we have to handle the circumstance that we only sample from $U \cup L$ and ignore $X$.

    There is a further issue regarding $U$.
    Let $s$ be the smallest index such that $U_s$ and $c$ intersect, \ie, $U_s$ is the first set with binary entries in column $j$.
    Then rows sampled for $U_s$ may be located outside of $c$ at positions with wildcards in column $j$.
    We avoid the complications caused by the wildcards by only considering classes $U_{i}$ for $i > s$.

    To summarize, $c$ has at least $\varepsilon r$ selected entries and we ignore at most $2 \varepsilon^2 r$ of these due to $X$ and $U_s$.
    For each $i > s$, we sample $1/\varepsilon^3$ rows from $U_i$. 
    Let $c'$ be $c$ restricted to $\bigcup_{i > s} U_{i}$ and let $c''$ be $c$ restricted to  $L$.
    Let $\hat{c}$ be $c$ without $U_s$ and $X$ and let $\bar{c}$ be the part of $c$ in $X \cup U_s$.
    For each $i$, let $c'_i$ be the fraction of zeros of $c'$ in $U_i$ and $c''_i$ the fraction of zeros of $c''$ in $L_i$.

    For each $i \le \ell$, we define $p'_i$ to be the fraction of zeros $c'_i$ and $p''_i$ the fraction of zeros $c''_i$.

    We define a random variables $Y'_{i,k}$ for each $s < i \le 1/\varepsilon^2$ and $Y''_{i,k}$ for each $1 \le i \le 1/\varepsilon^2$. 
    In both cases, $1 \le k \le 1/\varepsilon^3$.
    For each $i,k$, we pick an entry from $c'_{i}$ ($c''_i$) uniformly at random. 
    Then $Y'_{i,k}$ ($Y''_{i,k}$) is the value of the picked entry.
    For all $i,k$, $E[Y'_{i,k}] = 1 - p'_{i}$ and $E[Y''_{i,k}] = 1 - p''_i$. 
    Observe that the $Y_{i,k}$ are independent Poisson trials.
    Let $Y' := \sum_{s<i ,1 \le k \le 1/\varepsilon^3} Y'_{i,k}$ 
    and $Y'':= \sum_{i, 1 \le k \le 1/\varepsilon^3} Y''_{i,k}$.
    We want to use Chernoff bounds to control the probability to take the wrong decision.
    It is sufficient to consider $Y'$ with $s = 1/\varepsilon^2-1$, since in all other cases the probabilities are amplified more.
    Observe that we do not have to consider smaller $s$ because we are given a good SWC-instance and therefore there are no wildcards in $L$ or $L'$.

    Let $\mu' := E[Y']$.
    We analyze the ranges of $\mu'$ separately.
    \subparagraph{Case 1:} Let us assume that $\mu' \in [0,1/(2\euler\varepsilon^3)]$.
    We define $\delta' := 1/(2\mu'\varepsilon^3) - 1$.
    Using a multiplicative Chernoff bound (cf.~\cite{MU05_probability}), we obtain
    \begin{align}
        \Pr(Y' \ge 1/(2\varepsilon^3)) &< \Bigl(\frac{\euler^{\delta'}}{(1+\delta')^{(1+\delta')}}\Bigr)^{\mu'} = \Bigl(\frac{1}{1 + \delta'}\Bigr)^{\mu'} \Bigl(\frac{\euler}{{1 + \delta'}}\Bigr)^{\mu'\delta'}\\
        &= (2 \mu' \varepsilon^3)^{\mu'}(\euler \cdot 2\mu'\varepsilon^3)^{(1/(2\varepsilon^3)  - \mu')}\label{eq:rhs}
    \end{align}
    Note that both terms of \eqref{eq:rhs} are numbers between zero and one.
    If $\mu' < 1/\varepsilon$, the right term is smaller than $\varepsilon^4 \mu'$. 
    Otherwise the left term is smaller than $\varepsilon^4 \mu'$

    The range of $\mu'$ implies that the majority of entries in $\hat{c}'$ is zero.
    Recall that $\hat{c}'$ has an $\varepsilon^3 \mu'$ fraction of zeros.
    The expected number of errors done by the algorithm is therefore at most
    $(1 - \varepsilon^4 \cdot \mu') \cdot (\varepsilon^3 \mu') + \varepsilon^4 \cdot \mu' \cdot (1 - \varepsilon^3 \mu') = (1+\varepsilon) \varepsilon^3 \mu'$.

    \subparagraph{Case 2:} Let us assume that $\mu' \in (1/(2\euler\varepsilon^3), 1/(2\varepsilon^3) - 1/\varepsilon^2]$.
    We use Hoeffding's inequality \cite{Hoe63_probability} to analyze the range.
    To this end, we scale $Y'$ and obtain $\bar{Y}':= \varepsilon^3 Y'$, which has values between zero and one.
    Then 
    \[
        \Pr(\bar{Y}' - E[\bar{Y}'] \ge \varepsilon) \le \euler^{-2\varepsilon^2/\varepsilon^3} = \euler^{-2/\varepsilon}\,.
    \]

    Since for sufficiently small $\varepsilon$, $\euler^{-2/\varepsilon} < \varepsilon/(2\euler) \le \varepsilon^4 \mu'$,
    again we obtain a $(1+\varepsilon)$-approximation in expectation.

    All other ranges now follow immediately:
    For $\mu' \in (1/(2\varepsilon^3) - 1/\varepsilon^2,1/(2\varepsilon^3)]$ every solution is a $(1+O(\varepsilon))$-approximation and for larger $\mu'$ the majority of entries in $\hat{c}'$ is one.
    The analysis is analogous.

    In order to combine $Y'$ and $Y''$, we introduce a bias for $Y'$ such that we count rows $i$ for $s < i \le \ell$ with a factor $(1-\varepsilon)/(\varepsilon - \varepsilon^2)$.
    Then
    \[
        \bar{Y} := \frac{\bar{Y}' \cdot (\ell-s)(1-\varepsilon)/(\varepsilon - \varepsilon^2) + \bar{Y}'' \cdot \ell}{(\ell-s)(1-\varepsilon)/(\varepsilon - \varepsilon^2) + \ell}\,.
    \]
    Then, using the union bound, setting $\sigma_j = 0$ for $\bar{Y} < 1/2$ and $\sigma_j = 1$ otherwise gives an expected $1+O(\varepsilon)$ approximation within $\hat{c}$.
    Errors in $\bar{c}$ are either also errors in an optimal solution, or they contribute at most a factor $O(\varepsilon)$ to the total number of errors.
    Thus overall we obtain an approximation ratio $1+O(\varepsilon)$ within $c$.
    The algorithm $\SWC$ has at most the same approximation ratio, since the only difference is that we do not fix the $Y_{i,k}$ to be zero or one.
    Thus the random process used by the algorithm can only have a lower variance.

    This finishes our analysis for $\tau(M)_{*,j}$.
    For $\tau'(M)_{*,j}$, the proof is analogous.

\end{proof}

We introduced a small but easy to handle imprecision due to the assumption that we can choose exactly the same number of strings from each range.

\section{A simplified DP for a single solution string.}\label{sec:single}
We describe a dynamic program (DP) for a simplified setup with SWC-instances that consists of strings only from one of the two solution strings and the DP computes a single solution string.

\subparagraph{Algorithm ($\textsc{SWC}^{\sigma}$).}
We first globally guess the value $|\tau(M)| =: r$, i.e., we run the algorithm for all possible values and keep the best outcome.
The algorithm works in two phases. 
The first phase is an initialization.

We initialize \emph{each} of cell $\zeta := D(B,C,T)$ with the value computed by $\textsc{SWC}_{\varepsilon^3}$ with the following parameters.
As $U_i$ and $L_i$, we use the chunks $C$.
Since we only consider one solution string, we do not have to fix $r'$ or $U'_i$.
In the execution of $\SWC$, we use the selection $T$ instead of trying all possible selections, i.e., $T$ determines all $\tilde{U}_i$ and $\tilde{L}_i$ in the algorithm.

The value of $\zeta$ is the number of errors in $B$.
The computed solution is $\sigma_\zeta$.
We update the cells the second phase as follows. 
Consider a DP cells $\zeta = D(B,C,T)$ and let $\Pi$ be the set of all possible predecessors of $\zeta$.
Suppose that all cells in $\Pi$ are updated already. 
(This is the case, if we consider cells ordered by increasing value $a$, breaking ties arbitrarily.)

We try all cells $\hat \zeta \in \Pi$ (with all of its parameters marked by $\hat{\cdot}$\,) and consider the block $B$ from column $\en_{\hat{b} +1}$ on, which we call $\tilde{B}$.
We then run $\SWC$ with the parameters and selections from $\zeta$ for $\tilde{B}$.
Let $\mathrm{err}$ be the number of errors of the solution in the rows $a$ to $c-1$ of $\tilde{B}$.
We concatenate the computed solution string to $\sigma_{\hat{\zeta}}$.
The new value of $\zeta$ is $\min\{\zeta, \hat \zeta + \mathrm{err}\}$.
Overall, the value of $\zeta$ is the minimum value over all $\hat \zeta \in \Pi$.

We iterate this procedure until all cells are updated.
It might happen, however, that we were not able to compute the entire solution yet.
The reason is that valid DP cells as specified select a large number of rows, which may not be possible in the end.
In order to finish the DP, we additionally consider special cells that are defined as before, but with $c=n$. 
Intuitively, we use these cells when only at most $1/\varepsilon^4$ rows of $\tau(M)$ are left.
For these cells, our computation considers the optimal solution for the suffix of $\sigma$.

\begin{lemma}\label{lem:simpleDP}
    For SWC-instances $M$ of $\GMEC$ with a restriction that $M$ contains strings from only one of the two solution strings ($\sigma$ or $\sigma'$), the above algorithm is a PTAS.
\end{lemma}
\begin{proof} 
    Since all binary strings are feasible solutions, our algorithm vacuously produces a valid solution. 
    The number of different DP cells is polynomial in the instance size since the number of variables is a constant (depending on $\varepsilon$) and each variable has a polynomial range.
    All computations can be done in polynomial time. Therefore the overall running time of the algorithm is polynomial.

    To analyze the quality of the computed solution, we partition $\tau(M)$ into ranges.
    Starting from the top-most row of $\tau(M)$, for each $i \ge 0$, the $i$th range $Y_i$ contains the next 
    $(\varepsilon^{2i} - \varepsilon^{(2i+2)})r$ rows of $\tau(M)$.
    To be consistent with properties needed in later proofs, we ensure that the first row of each $Y_i$ is contained in $\tau(M)$ and thus
    we add the rows between $Y_i$ and $Y_{i+1}$ to $Y_i$.
    We note that if only a constant number of rows of $\sigma(M)$ are left, we can compute the partial solutions optimally and
    there are DP cells for exactly this purpose:
    there is a DP cell $\zeta_i$ such that the last at most $1/\varepsilon^4$ rows of $\tau(M)$ are located between $a$ and $c$ and $Y_i$ contains exactly these rows.
    To keep a clean notation, in the following we implicitly assume that cells with constantly many rows of $\sigma(M)$ are handled separately.

    The block $B_0$ contains the rows of $Y_0$ and the columns one to the end of the first row of $\tau(Y_0)$. 
    For each $i > 0$, block $B_i$ contains the rows of $Y_i$ and $Y_{i+1}$. It contains the columns after those of $B_{i-1}$ to the end of the first row of $Y_i$.

    According to Definition~\ref{def:dp-cell-one-solution-string}, the remaining parameters for cells $\zeta_i$ lead to at least as good a solution as the following choice.
    The set $C$ is chosen such that each block is the $U$ and $L$ part of an $\varepsilon^2$-trisection and the chunks are the subdivisions of the trisection (Definitions~\ref{def:trisection} and \ref{def:subdivision}).
    The selections $S$ are chosen in the same way as $\SWC$ would choose them. 

    We inductively show that the value of each $\zeta_i$ is at most a factor $(1+\varepsilon)$ larger than the number of errors of an optimal solution restricted to the considered prefix.
    For $i=0$ we only consider $B_0$ and the invariant follows directly from Lemma~\ref{lem:swc-gap}.

    Suppose now that $i \ge 0$ and for all $\tilde{i}<i$ the invariant is true.
    Then we consider $\zeta_i$. Let $\tilde{B}_i$ be the part of $B_i$ after $B_{i-1}$.

    We apply Lemma~\ref{lem:swc-gap} to compute the string $\sigma^{\zeta_i}$.
    We obtain a $(1+\varepsilon)$ for the prefix covered by $\sigma^{\zeta_i}$ for the following reason.
    The part before $\tilde{B}_i$ was fixed, and by our induction hypothesis, independent of the rows considered in $\tilde{B}$ we already have a $(1+\varepsilon)$ approximation.
    The part of $\sigma^{\zeta_i}$ within $\tilde{B}_i$ gives a $(1+\varepsilon)$ approximation by the claim of Lemma~\ref{lem:swc-gap}.

    We continue the induction until the entire string $\sigma$ is determined.
\end{proof}

\end{document}